\documentclass[12pt,a4paper]{article} 
\usepackage[left=38mm,right=28mm, top=35mm, bottom=35mm]{geometry}
\usepackage[english]{babel}
\usepackage[babel,english=american]{csquotes}
\usepackage{color}
\usepackage{amssymb,amsfonts,amsmath,amsthm,bbm}
\usepackage{latexsym}
\usepackage{graphicx}
\usepackage{multirow}
\usepackage{diagbox}
\usepackage{stmaryrd}
\usepackage{dsfont}
\usepackage{nicefrac}
\usepackage{enumerate}
\usepackage{mathrsfs}
\usepackage{titlesec}
\usepackage{natbib}
\usepackage{subfigure}
\usepackage{bm}
\usepackage{ushort}
\usepackage[reftex]{theoremref}

\newcommand{\R}{\mathbb{R}}

\newcommand{\Q}{\mathbb{Q}}
\newcommand{\N}{\mathbb{N}}

\newcommand{\Linf}[2]{L_{#2}^\infty(#1)}
\newcommand{\SLinf}[2]{\mathscr{L}_{#2}^\infty(#1)}

\newcommand{\CB}{\mathcal{B}}

\newcommand{\CF}{\mathcal{F}}
\newcommand{\CG}{\mathcal{G}}
\newcommand{\CH}{\mathcal{H}}

\newcommand{\CP}{\mathcal{P}}

\newcommand{\CX}{\mathcal{X}}

\newcommand{\PW}{\mathds{P}}
\newcommand{\PWQ}{\mathds{Q}}

\newcommand{\WR}{\Omega}
\newcommand{\om}{\omega}

\newcommand{\dLG}{\widetilde{\Lambda}_{\CG}}
\newcommand{\dL}{\widetilde{\Lambda}}
\newcommand{\cL}{\Lambda}
\newcommand{\rG}{\rho_{\CG}}
\newcommand{\rGn}{\eta_{\CG}}
\newcommand{\crg}[2]{\rG\left(#1,#2\right)}

\newcommand{\drG}{\widetilde{\rho}_\CG}
\newcommand{\drg}[2]{\drG\left(#1,#2\right)}
\newcommand{\rgn}[2]{\rGn\left(#1,#2\right)}
\newcommand{\dLg}[2]{\dLG\left(#1,#2\right)}
\newcommand{\eqdLg}[2]{\eq{\dL}_\CG\left(#1,#2\right)}
\newcommand{\cLG}{\Lambda_\CG}
\newcommand{\cLg}[1]{\cLG\left(#1\right)}

\newcommand{\vecone}{\mathbf{1}_d}
\newcommand{\veczero}{\mathbf{0}_d}

\newcommand{\BPi}{\mathbf{\Pi}}

\newcommand{\HLG}{\widehat{\Lambda}_\CG}
\newcommand{\HLg}[2]{\HLG\left(#1,#2\right)}
\newcommand{\HX}{\widehat{X}}
\newcommand{\HY}{\widehat{Y}}

\newcommand{\xo}{\overline{x}}

\newcommand{\ind}{\mathbbmss{1}}
\newcommand{\Xh}[1]{X^{(#1)}}
\newcommand{\EW}[2]{\mathds{E}_{#1}\left[#2\right]}
\newcommand{\BEW}[3]{\mathds{E}_{#1}\left[\left.#2\,\right|\,#3\right]}
\newcommand{\eq}[1]{#1}
\newcommand{\eqind}{\eq{\ind}}
\newcommand{\eqalpha}{\eq{\alpha}}
\newcommand{\eqlambda}{\eq{\lambda}}
\newcommand{\eqF}{\eq{F}}
\newcommand{\eqFom}{F(\om)}
\newcommand{\eqX}{\eq{X}}

\newcommand{\eqXom}{X(\om)}
\newcommand{\eqXh}[1]{\eq{X}^{(#1)}}
\newcommand{\eqXomh}[1]{X^{(#1)}(\om)}
\newcommand{\eqx}{\eq{x}}
\newcommand{\eqHX}{\eq{\HX}}
\newcommand{\eqG}{\eq{G}}
\newcommand{\eqGom}{G(\om)}
\newcommand{\eqY}{\eq{Y}}
\newcommand{\eqYom}{Y(\om)}
\newcommand{\eqy}{\eq{y}}
\newcommand{\eqHY}{\eq{\HY}}
\newcommand{\eqH}{\eq{H}}
\newcommand{\eqHom}{H(\om)}
\newcommand{\eqZ}{\eq{Z}}
\newcommand{\eqZom}{Z(\om)}
\newcommand{\eqZh}[1]{\eq{Z}^{(#1)}}
\newcommand{\eqZhom}[1]{Z^{(#1)}(\om)}

\newcommand{\eqone}{\eq{1}}
\newcommand{\eqvecone}{\eq{\mathbf{1}}_d}

\DeclareMathOperator{\AVaR}{AVaR}

\DeclareMathOperator*{\essinf}{essinf}
\DeclareMathOperator*{\esssup}{esssup}

\DeclareMathOperator{\Imag}{Im}

\DeclareMathOperator{\VaR}{VaR}

\newcommand{\supnorm}[1]{\left\|#1\right\|_{\infty}}

\numberwithin{figure}{section}
\numberwithin{table}{section}
\numberwithin{equation}{section}

\newtheorem{Theorem}{Theorem}[section]
\newtheorem{Proposition}[Theorem]{Proposition}
\newtheorem{Lemma}[Theorem]{Lemma}

\newtheorem{Definition}[Theorem]{Definition}

\theoremstyle{definition}
\newtheorem{Example}[Theorem]{Example}

\theoremstyle{remark}
\newtheorem{Remark}[Theorem]{Remark}

\begin{document}
\selectlanguage{english}
\title{Risk-Consistent Conditional Systemic Risk Measures}
\author{Hannes Hoffmann\thanks{Department of Mathematics, University of Munich, Theresienstra{\ss}e 39, 80333 Munich, Germany. Emails: hannes.hoffmann@math.lmu.de, meyer-brandis@math.lmu.de and gregor.svindland@math.lmu.de.}\and
Thilo Meyer-Brandis\footnotemark[1]\and
Gregor Svindland\footnotemark[1]}
\maketitle
\begin{abstract}
We axiomatically introduce {\em risk-consistent conditional systemic risk measures} defined on multidimensional risks. This class consists of those conditional systemic risk measures which can be decomposed into a state-wise conditional aggregation and a univariate conditional risk measure. Our studies extend known results for unconditional risk measures on finite state spaces. We argue in favor of a conditional framework on general probability spaces for assessing systemic risk.
Mathematically, the problem reduces to selecting a realization of a random field with suitable properties. Moreover, our approach covers many prominent examples of systemic risk measures from the literature and used in practice. 
\\[4mm]\noindent\textbf{Keywords:} conditional systemic risk measure, conditional aggregation, risk-consistent properties, conditional value at risk, conditional expected short fall.  
\end{abstract}

\section{Introduction}
The recent financial crisis revealed weaknesses in the financial regulatory framework when it comes to the protection against systemic events. Before, it was generally accepted to measure the risk of financial institutions on a stand alone basis. In the aftermath of the financial crisis risk assessment of financial systems as well as their impact on the real economy has become increasingly important, as is documented by a rapidly growing literature; see e.g.\ \cite{Amini2013b} or \cite{Bisias2012} for a survey and the references therein. Parts of this literature are concerned with designing appropriate risk measures for financial systems, so-called systemic risk measures.
The aim of this paper is to axiomatically characterize the class of systemic risk measures $\rho$ which admit a decomposition of the following form:
\begin{equation}\label{eq:aim1}
\rho(\eqX)=\eta\left(\cL(\eqX)\right),
\end{equation}
where $\cL$ is a state-wise aggregation function over the $d$-dimensional random risk factors $\eqX$ of the financial system, e.g.\ profits and losses at a given future time horizon, and $\eta$ is a univariate risk measure.
The aggregation function determines how much a single risk factor contributes to the total risk $\cL(\eqX)$ of the financial system in every single state, whereas the so-called base risk measure $\eta$ quantifies the risk of $\cL(\eqX)$. \cite{Chen2013} first introduced axioms for systemic risk measures, and showed that these admit a decomposition of type \eqref{eq:aim1}. Their studies relied on a finite state space and were carried out in an unconditional framework. \cite{Kromer2013} extend this to arbitrary probability spaces, but keep the unconditional setting. The main contributions of this paper are: 
\begin{itemize}
\item[1.] We axiomatically characterize systemic risk measures of type \eqref{eq:aim1} in a conditional framework, in particular we consider conditional aggregation functions and conditional base risk measures in \eqref{eq:aim1}.
\item[2.] We allow for a very general structure of the aggregation, which is flexible enough to cover examples from the literature which could not be handled in axiomatic approaches to systemic risk so far. 
\item[3.] We work in a less restrictive axiomatic setting, which gives us the flexibility to study systemic risk measures which for instance need not necessarily be convex or quasi-convex, etc. This again provides enough flexibility to cover a vast amount of systemic risk measures applied in practice or proposed in the literature. It also allows us to identify the relation between properties of $\rho$ and properties of $\cL$ and $\eta$, and in particular the mechanisms behind the transfer of properties from $\rho$ to $\cL$ and $\eta$, and vice versa. This is related to the following point 4.    
\item[4.] We identify the underlying structure of the decomposition \eqref{eq:aim1} by defining systemic risk measures solely in terms of so called risk-consistent properties and properties on constants.
\end{itemize}
In the following we will elaborate on the points 1.--4.\ above.

\paragraph{1. A conditional framework for assessing systemic risk}~

\smallskip\noindent
We consider systemic risk in a conditional framework.
The standard motivation for considering conditional univariate risk measures (see e.g.\ \cite{Detlefsen2005} and \cite{Acciaio2011}) is the conditioning in time, and the argumentation in favor of this also carries over to multivariate risk measures. 
However, apart from a dynamic assessment of the risk of a financial system, it might be particularly interesting to consider conditioning in space. In that respect \cite{Follmer2014a} recently introduced and studied so-called spatial risk measures for univariate risks.
Typical examples of spatial conditioning are conditioning on events representing the whole financial system or parts of that system, such as single financial institutions, in distress. This is done to study the impact of such a distress on (parts of) the financial system or the real economy, and thereby to identify systemically relevant structures. For instance the Conditional Value at Risk (CoVaR) introduced in \cite{Adrian2011} considers for $q\in (0,1)$ the $q$-quantile of the distribution of the netted profits/losses of a financial system $\eqX=(\eqX_1,\ldots, \eqX_d)$ conditional on a crisis event $C(X_i)$ of institution $i$:
\begin{equation}\label{eq:covar:intro}\PW\left(\left. \sum_{i=1}^d X_i\leq -\text{CoVaR}_q(\eqX) \right|C(X_i)\right)=q; \end{equation} see \thref{ex:COVAR}.
More examples can be found in \cite{Cont2013}, \cite{Engle2014}, \cite{Acharya2010}. Such risk measures fit naturally in a conditional framework; cf.\ \thref{ex:COVAR} and \thref{ex:DIP}. 



\paragraph{2. Aggregation of multidimensional risk}~

\smallskip\noindent
A quite common aggregation rule for a multivariate risk $\eqX=(\eqX_1,\ldots, \eqX_d)$ is simply the sum 
$$\Lambda_{\text{sum}}(X)=\sum_{i=1}^d X_i; $$ see the definition of CoVaR in \eqref{eq:covar:intro}.
$\Lambda_{\text{sum}}(X)$ represents the total profit/loss after the netting of all single profits/losses. However, such an aggregation rule might not always be reasonable when measuring systemic risk. 
The major drawbacks of this aggregation function in the context of financial systems are that profits can be transferred from one institution to another and that losses of a financial institution cannot trigger additional contagion effects.  
%
%
Those deficiencies are overcome by aggregation functions which explicitly account for contagion effects within a financial system. For instance, based on the approach in \cite{Eisenberg2001}, the authors in \cite{Chen2013} introduce such an aggregation rule which however, due to the more restrictive axiomatic setting, exhibits the unrealistic feature that in case of a clearing of the system institutions might decrease their liabilities by more than their total debt level. We will present a more realistic extension of this contagion model together with a small simulation study in \thref{EX:CM}. 

Moreover, we present reasonable aggregation functions which are not comprised by the axiomatic framework of \cite{Chen2013} or \cite{Kromer2013}. In particular this includes \emph{conditional aggregation functions} which come naturally into play in our framework; see \thref{EX:Stochasticdiscount}.


\paragraph{3.--4. Axioms for systemic risk measures}~

\smallskip\noindent
Our aim is to identify the relation between properties of $\rho$ and properties of $\cL$ and $\eta$ in \eqref{eq:aim1} respectively, and in particular the mechanisms behind the transfer of properties from $\rho$ to $\cL$ and $\eta$, and vice versa. We will show that this leads to two different classes of axioms for conditional systemic risk measures. One class concerns the behavior on deterministic risks, so-called properties on constants. The other class of axioms ensures a consistency between state-wise and global - in the sense of over all states - risk assessment. This latter class will be called risk-consistent properties. 
%

The risk-consistent properties ensure a consistency between local - that is $\omega$-wise - risk assessment and the measured global risk. For example, \emph{risk-antitonicity} is expressed by: if for given risk vectors $\eqX$ and $\eqY$ it holds that $\rho(X(\omega)) \geq \rho(Y(\omega))$  in almost all states $\om$, then $\rho(\eqX) \geq \rho(\eqY)$. The naming \emph{risk-antitonicity}, and analogously the naming for the other risk-consistent properties, is motivated by the fact that antitonicity is considered with respect to the order relation $\rho(X(\om)) \geq \rho(Y(\om))$ induced by the $\omega$-wise risk comparison of two positions and not with respect to the usual order relation on the space of random vectors.

Note that for a univariate risk measure $\rho$ which is constant on constants, i.e.\ $\rho(\eq{x})=-x$ for all $x\in\R$, risk-antitonicity is equivalent to the 'classical' antitonicity with respect to the usual order relation on the underlying space of random variables. In a general multivariate setting this equivalence does not hold anymore. However, we will show that properties on constants in conjunction with corresponding risk-consistent properties imply the classical properties on the space of risks. This makes our risk model very flexible, since we may identify systemic risk measures where for example the corresponding aggregation function $\cL$ in \eqref{eq:aim1} is concave, but the base risk measure $\eta$ is not convex. Moreover, it will turn out that the properties on constants basically determine the underlying aggregation rule in the systemic risk assessment, whereas the risk-consistent properties translate to properties of the base risk measure in the decomposition \eqref{eq:aim1}.  

Some of the risk-consistent properties, however partly under different names, also appear in the frameworks of \cite{Chen2013} and \cite{Kromer2013}. For instance what we will call risk-antitonicity is called preference consistency in \cite{Chen2013}. In our framework we emphasize the link between the risk-consistent properties (and the properties on constants) and the decomposition \eqref{eq:aim1}. This aspect has not been clearly worked out so far. It leads us to introducing a number of new axioms and to classifying all axioms within the mentioned classes of risk-consistent properties and properties on constants. 

\paragraph{Structure of the paper}~

\smallskip\noindent
In Section \ref{SEC:DECOMP} we introduce our notation and the main objects of this paper, that is the risk-consistent conditional systemic risk measures, the conditional aggregation functions and the conditional base risk measures as well as their various extensions. At the end of Section \ref{SEC:DECOMP} we state our main decomposition result (\thref{T1}) for risk-consistent conditional systemic risk measures. Moreover, \thref{T2} reveals the connection between risk-consistent properties and properties on constants on the one hand and the classical properties of risk measures on the other hand. Section~\ref{SEC:PROOF} is devoted to the proofs of \thref{T1} and \thref{T2}. In Section~\ref{SEC:EX} we collect our examples.


\section{Decomposition of systemic risk measures}\label{SEC:DECOMP}
Throughout this paper let $(\WR,\CF,\PW)$ be a probability space and $\CG$ be a sub-$\sigma$-algebra of $\CF$. 
$\SLinf{\CF}{}:=\mathscr{L}^\infty(\WR,\CF,\PW)$ refers to the space of $\CF$-measurable, $\PW$-almost surely (a.s.) bounded random variables and $\SLinf{\CF}{d}$ to the $d$-fold cartesian product of $\SLinf{\CF}{}$.
As usual, $L^\infty(\CF)$ and $L^\infty_d(\CF)$ denote the corresponding spaces of random variables/vectors modulo $\PW$-a.s.\ equality.
For $\CG$-measurable random variables/vectors analogue notations are used.

\noindent In general, upper case letters will represent random variables, where $X,Y,Z$ are multidimensional and $F,G,H$ are one-dimensional, and lower case letters deterministic values.

\noindent We will use the usual componentwise orderings on $\R^d$ and $\Linf{\CF}{d}$, i.e.\ $x=(x_1,\ldots,x_d)\geq y=(y_1,\ldots,y_d)$ for $x,y\in \R^d$ if and only if $x_i\geq y_i$ for all $i=1,\ldots, d$, and similarly $\eqX\geq\eqY$ if and only if $X_i\geq Y_i$ a.s.\ for all $i=1,...,d$.
Furthermore $\vecone$ and $\veczero$ denote the $d$-dimensional vectors whose entries are all equal to 1 or all equal to 0, respectively. 

\noindent When deriving our main results we will run into similar problems as one faces in the study of stochastic processes: At some point it will not be sufficient to work on equivalence classes, but we will need a specific nice realization or version of the process, for instance a version with continuous paths, etc.   
In the following, by a realization of a function $\rG:L^\infty_d(\CF)\to L^\infty(\CG)$ we mean a selection of one representative in the equivalence class $\rG(\eqX)$ for each $\eqX\in L^\infty_d(\CF)$, i.e.\ a function $\rG(\cdot,\cdot):\Linf{\CF}{d}\times\WR\to\R$ where $\rG(\eqX,\cdot)\in\SLinf{\CG}{}$ with $\rG(\eqX,\cdot)\in \rG(\eqX)$ for all $\eqX\in\Linf{\CF}{d}$. We emphasize that in the following we will always denote a realization of a function $\rG$ by its explicit dependence on the two arguments: $\rG(\cdot,\cdot)$.
Indeed, our decomposition result in \thref{T1} will be based on the idea to break down a random variable into every single scenario and evaluating it separately. This implies working with appropriate realizations which will satisfy properties which we will denote \textit{risk-consistent} properties. 

Also for risk factors we will work both with equivalence classes of random vectors in $\Linf{\CF}{d}$  and their corresponding representatives in $\SLinf{\CF}{d}$.
However, in contrast to the realizations of $\rG$ introduced above, here the considerations do not depend on the specific choice of the representative. 
Hence for risk factors $X\in\Linf{\CF}{d}$ we will stick to usual abuse of notation of also writing $X$ for an arbitrary representative in $\SLinf{\CF}{d}$ of the corresponding equivalence class. This will become clear from the context. In particular, $X(\om)$ denotes an arbitrary representative of the corresponding equivalence class evaluated in the state $\om\in\WR$.

Finally, we write $x\in\R^d$ both for real numbers and for (equivalence classes of) constant random variables depending on the context.

The following definition introduces our main object of interest in this paper:

\begin{Definition}[Risk-consistent Conditional Systemic Risk Measure]\th\label{DEFCSRM}~\\
A function $\rG: \Linf{\CF}{d}\to \Linf{\CG}{}$ is called a \emph{risk-consistent conditional systemic risk measure} (CSRM), if it is 
\begin{description}
\item[\emph{Antitone on constants:}] For all $x,y\in\R^d$ with $x\geq y$ we have $\rG(\eqx)\leq \rG(\eqy)\,,$
\end{description}
and if there exists a realization $\crg{\cdot}{\cdot}$ such that the restriction 
	\begin{equation}\label{eq:rev1}\drG:\R^d\times\WR\to\R;~ x\mapsto\crg{\eqx}{\om}\end{equation}
	 has \emph{continuous paths}, i.e.\ $\drG$ is continuous in its first argument a.s., and it satisfies 
\begin{description}
\item[\emph{Risk-antitonicity:}] For all $\eqX,\eqY\in\Linf{\CF}{d}$ with 
$\drg{\eqXom}{\om}\geq\drg{\eqYom}{\om}$ a.s.\ we have $\rG(\eqX)\geq\rG(\eqY)$. 
\end{description}
\smallskip\noindent Furthermore, we will consider the following properties of $\rG$ on constants: 
\begin{description}
\item[\emph{Convexity on constants:}] $\rG\left(\eqlambda \eqx+(\eqone-\eqlambda)\eqy\right)\leq \eqlambda \rG(\eqx)+(\eqone-\eqlambda)\rG(\eqy)$ for all constants $x,y\in \R^d$ and $\lambda\in[0,1]$;
\item[\emph{Positive homogeneity on constants:}] $\rG(\eqlambda \eqx)=\eqlambda\rG(\eqx)$ for all $x\in\R^d$ and $\lambda\geq0$.
\end{description}
We will also consider the following risk-consistent properties of $\rG$:
\begin{description}
\item[\emph{Risk-convexity:}] If for $\eqX,\eqY,\eqZ\in\Linf{\CF}{d}$ there exists an $\eqalpha\in\Linf{\CG}{}$ with $0\leq\alpha\leq 1$ such that
$\drg{\eqZom}{\om}=\alpha(\om)\drg{\eqXom}{\om}+\big(1-\alpha(\om)\big)\drg{\eqYom}{\om}$ a.s.,
then  $\rG(\eqZ)\leq\eqalpha\rG(\eqX)+(\eqone-\eqalpha)\rG(\eqY)$;
\item[\emph{Risk-quasiconvexity:}] If for $\eqX,\eqY,\eqZ\in\Linf{\CF}{d}$ there exists an $\eqalpha\in\Linf{\CG}{}$ with $0\leq\alpha\leq 1$ such that 
$\drg{\eqZom}{\om}=\alpha(\om)\drg{\eqXom}{\om}+\big(1-\alpha(\om)\big)\drg{\eqYom}{\om}$ a.s.,
then  $ \rG(\eqZ)\leq\rG(\eqX)\vee\rG(\eqY)$;
\item[\emph{Risk-positive homogeneity:}] If for $\eqX,\eqY\in\Linf{\CF}{d}$ there exists an $\eqalpha\in\Linf{\CG}{}$ with $\alpha\geq 0$ such that
$\drg{\eqYom}{\om}=\alpha(\om)\drg{\eqXom}{\om}$  a.s.,
then  $ \rG(\eqY)=\eqalpha\rG(\eqX)$;
\item[\emph{Risk-regularity:}]  $\crg{\eqX}{\om}=\drg{\eqXom}{\om}$ a.s.\ for all $\eqX\in\Linf{\CG}{d}$.
\end{description}
\end{Definition}
We will see in \thref{T1} that risk-antitonicity is the crucial property which guarantees that $\rG$ allows a conditional decomposition analogously to \eqref{eq:aim1}.
The idea behind all risk-consistent properties is that they ensure a consistency between local - that is $\omega$-wise - risk assessment and the measured global risk. 
Consider for instance again the risk-antitonicity property and suppose we are given an event $A\in\CG$ and random risk factors $\eqZ\in \Linf{\CF}{d}$ as well as $\eqX,\eqY\in \Linf{\CF}{d}$ such that on the level of our realization which satisfies the risk-antitonicity we have $\drg{\eqXom}{\om}\geq \drg{\eqYom}{\om}$ a.s.\ on $A$. In other words for almost all $\om\in A$, the risk of the constant risk factors $\eqXom$ evaluated in $\om$ is higher than the corresponding risk of $\eqYom$ evaluated in $\om$. Now consider the modified risk factors $\eqZ_X:=\eqX\eqind_A+\eqZ\eqind_{A^C}$ and $\eqZ_Y:=\eqY\eqind_A+\eqZ\eqind_{A^C}$ where we modify $\eqZ$ on $A$ in such a way that $\eqZ_Y$ is preferred on almost every state in $A$ to $\eqZ_X$, and 
otherwise both risk factors are identical. Then risk-antitonicity implies that $\rG(\eqZ_Y)\leq \rG(\eqZ_X)$. 

Our definition of a CSRM is based on properties on constants together with risk-consistent properties. It turns out (see \thref{T1}) that the properties on constants translate into the corresponding properties of the (conditional) aggregation function and the risk-consistent properties translate into the corresponding properties of the (conditional) base risk measure in the decomposition of a CSRM. Moreover, a natural question is to which extend CSRM's also fulfill the established properties of risk measures in the literature. For instance, antitonicity on $L^\infty_d$, i.e.\ $\eqX\geq \eqY$ implies $\rG(\eqX)\leq \rG(\eqY)$, is commonly accepted as a minimal requirement for risk measures. Further, quasiconvexity or the stronger condition of convexity on $L^\infty_d$ are properties often asked for as they correspond to the requirement that diversification should not be penalized, cf.\ \cite{Cerreia-Vioglio2011}. Also, an important subclass are those CSRM which are positive homogeneous, as for example the 
CoVaR or the 
CoES introduced in \cite{Adrian2011}; see \thref{ex:COVAR} and \thref{ex:COES}. In general, it will turn out (see \thref{T2}) that properties on constants combined with the corresponding risk-consistent properties will imply properties such as antitonicity, (quasi-) convexity or positive homogeneity of $\rG$ on $L^\infty_d$. For example, antitonicity on constants in conjunction with risk-antitonicity implies antitonicity on $L^\infty_d$. 

One might ask in which setting it is possible to formulate the risk-consistent properties directly in terms of the function $\rG$ without requiring the existence of a particular realization of this function. As we will see in the next \thref{P1} this is possible if $\rho_\CG(\eqx)$ has a discrete structure for all $x\in\R^d$. For the sake of brevity we omit the proof. 
\begin{Proposition}\th\label{P1}
	Let $\rG:\Linf{\CF}{d}\to \Linf{\CG}{}$ be a function which has a realization with continuous paths. Further suppose that  
	\begin{equation}\label{gregor:eq:2}\rG(\eqx)=\sum_{i=1}^s \eq{a_i(x)\ind_{A_i}},\; x\in\R^d,\end{equation}
	where $a_i(x)\in\R$ and $A_i\in \CG$ are pairwise disjoint sets such that $\WR=\bigcup_{i=1}^s A_i$ for $s\in \N\cup\{\infty\}$.
	Define $k:\WR\to\N;\; \om\mapsto i \text{ such that }\om\in A_i.$
	Then $\rG$ is risk-antitone if and only if 
	\begin{equation}\label{hannes:discstrucra}\rG(\eqXom)\ind_{A_{k(\om)}}\geq \rG(\eqYom)\ind_{A_{k(\om)}}\; \text{a.s.\ implies }\rG(X)\geq\rG(Y),\end{equation} 
	where here the point evaluations $\eqXom,\eqYom\in\R^d$ have to be understood as equivalence classes of constant random variables.
	Also the remaining risk-consistent properties can be expressed in a similar way without requiring a particular realization of $\rG$.
\end{Proposition}

\begin{Remark}
Notice that in the setting of \thref{P1}, we had to require that there exists a realization with continuous paths. Sufficient criteria for $\rG$ which guarantee that such a continuous realizations exists are well known, e.g.\ Kolmogorov's criterion (see e.g.\ Theorem 2.1 in \cite{Revuz1999}). A sufficient specification of a CSRM solely in terms of $\rG$ (without employing any realization) is thus: if $\rG$ is antitone on constants, has a discrete structure \eqref{gregor:eq:2} and fulfills \eqref{hannes:discstrucra} and Kolmogorov's criterion, then $\rG$ is a CSRM.
\end{Remark}


In order to state our decomposition result we need to clarify what we mean by a (conditional) aggregation function and a conditional base risk measure. We start with the aggregation function.  

\begin{Definition}[Aggregation Functions]\th\label{def:agg}~\\
We call a function $\dL: \R^d\to \R$ a \emph{deterministic aggregation function} (DAF), if it has the following two properties: 
\begin{description}
\item[\emph{Isotonicity:}] If $x,y\in\R^d$ with $x \geq y$, then $\dL(x)\geq\dL(y)$;
\item[\emph{Continuity:}] $\dL$ is continuous.
\end{description}
A DAF is called concave or positive homogeneous, respectively, if it satisfies for all $x,y\in\R^d$
\begin{description}
\item[\emph{Concavity:}] If $\lambda\in[0,1]$, then $\dL\big(\lambda x+(1-\lambda) y\big)\geq \lambda\dL(x)+(1-\lambda)\dL(y)$;
\item[\emph{Positive homogeneity:}] $\dL(\lambda x)=\lambda\dL(x)$ for all $\lambda\geq0$.
\end{description}

\smallskip\noindent
Furthermore, a function $\dLG:\R^d\times\WR\to \R$ is a \emph{conditional aggregation function} (CAF), if 
\begin{enumerate}
\item\label{Lgmb} $\dLg{x}{\cdot}\in\SLinf{\CG}{}$ for all $x\in\R^d$,
\item\label{LgDAF} $\dLg{\cdot}{\om}$ is a DAF for all $\om\in\WR$.
\end{enumerate}
A CAF is called concave (positive homogeneous) if $\dLg{\cdot}{\om}$ is concave (positive homogeneous) for all  $\om\in\WR$.
\end{Definition}

\begin{Remark}\th\label{prodmb}
Note that, functions like CAFs which are continuous in one argument and measurable in the other also appear under the name of Carath\'eodory functions in the literature on differential equations. 
For Carath\'eodory functions it is well known (see e.g.\ \cite{Aubin2009} Lemma 8.2.6) that they are product measurable, i.e.\ every CAF $\dLG$ is $\CB(\R)\times\CG$-measurable.
\end{Remark}

Given a CAF $\dLG$, we extend the aggregation from deterministic to random vectors in the following way (which is well-defined due to \thref{prodmb} as well as isotonicity and property (i) in the definition of a CAF):
\begin{align} \label{def:AssCaf}
\cLG : & \ \ \Linf{\CF}{d} \to \Linf{\CF}{}, \quad 
   \eqX \mapsto \eqdLg{X(\om)}{\om} \,.
\end{align}
\begin{Remark}\th\label{hannes:pathwiseconcave}
Notice that the aggregation \eqref{def:AssCaf} of random vectors $X$ is $\om$-wise in the sense that given a certain state $\om\in \WR$, in that state we aggregate the sure payoff $X(\om)$. Consequently, properties such as isotonicity, concavity or positive homogeneity of the CAF $\dLG$ translate to the extended CAF $\cLG$. Hence, $\cLG$ always satisfies
\begin{equation}\label{hannes:CFisotone}
\cLg{\eqX}\geq\cLg{\eqY}\; \mbox{for all $\eqX,\eqY\in\Linf{\CF}{d}$ with $\eqX\geq\eqY$}.
\end{equation}
If $\dLG$ is concave, then for all $\eqX,\eqY\in\Linf{\CF}{d}$ and $\eqalpha\in\Linf{\CF}{}$ with $0\leq\alpha\leq1$ we have 
\begin{equation}\label{hannes:CFconcave}
\cLg{\eqalpha \eqX+(\eqone-\eqalpha) \eqY}\geq \eqalpha\cLg{\eqX}+(\eqone-\eqalpha)\cLg{\eqY},  
\end{equation}
and if $\dLG$ is positively homogeneous, then for all $\eqX\in\Linf{\CF}{d}$ and $\eqalpha\in\Linf{\CF}{}$ with $\alpha\geq0$:
\begin{equation}\label{hannes:CFposhomo}
\cLg{\eqalpha \eqX}= \eqalpha\cLg{\eqX}.
\end{equation}
\end{Remark}


The last yet undefined ingredient in our decomposition \eqref{eq:aim1} is the conditional base risk measure $\rGn$ which we define next. Notice that the domain $\CX$ of $\rGn$ depends on the underlying aggregation given by $\rG$. For example the aggregation function $\dL(x)=\sum_{i=1}^d \min\{x_i,0\},x\in\R^d$ only considers the losses. Hence, the corresponding base risk measure $\eta$ a priori only needs to be defined on the negative cone of $\Linf{\CF}{}$, even though it in many cases allows for an extension to $\Linf{\CF}{}$.
We will see in \thref{l1} that if $\CX$ is the image of an extended CAF $\cLG$ then $\CX$ is $\CG$-conditionally convex, i.e.\ $\eqF,\eqG\in\CX$ and $\eqalpha\in\Linf{\CG}{}$ with $0\leq\alpha\leq1$ implies $\eqalpha \eqF+(\eqone-\eqalpha)\eqG\in\CX$.

\begin{Definition}[Conditional Base Risk Measure]\th\label{def:CBRM}~\\
Let $\CX\subseteq\Linf{\CF}{}$ be a $\CG$-conditionally convex set.
A function $\rGn: \CX\to \Linf{\CG}{}$ is a \emph{conditional base risk measure} (CBRM), if it is
\begin{description}
\item[\emph{Antitone:}] $\eqF\geq \eqG$ implies $\rGn(\eqF)\leq\rGn(\eqG)$.
\end{description}
Moreover, we will also consider CBRM's which fulfill additionally one or more of the following properties:
\begin{description}
\item[\emph{Constant on constants:}] $\rGn(\eqalpha)=-\eqalpha$ for all $\eqalpha\in\CX\cap\Linf{\CG}{}$;
\item[\emph{Quasiconvexity:}]
$\rGn\left(\eqalpha \eqF+(\eqone-\eqalpha) \eqG\right)\leq \rGn(\eqF)\vee \rGn(\eqG)$ for all $\eqalpha\in\Linf{\CG}{}$ with $0\leq\alpha\leq1$;
\item[\emph{Convexity:}] $\rGn\left(\eqalpha \eqF+(\eqone-\eqalpha)\eqG\right)\leq\eqalpha\rGn(\eqF)+(\eqone-\eqalpha)\rGn(\eqG)$ for all $\eqalpha\in\Linf{\CG}{}$ with $0\leq\alpha\leq1$;
\item[\emph{Positive homogeneity:}] $\rGn(\eqalpha \eqF)=\eqalpha\rGn(\eqF)$ for all $\eqalpha\in \Linf{\CG}{}$ with $\alpha\geq0$ and $\eqalpha \eqF\in\CX$.
\end{description}
\end{Definition}
Constructing a CSRM by composing a CBRM and a CAF as in \eqref{eq:aim1}, we need a property for $\rGn$ which allows to 'extract' the CAF in order to obtain the properties on constants of $\rG$. The constant on constants property serves this purpose, but we will see in \thref{T1} that the following weaker property is also sufficient.
\begin{Definition}\th\label{def:ConstOnAggr}
A CBRM $\rGn:\CX\to\Linf{\CG}{}$ is called \emph{constant on a CAF $\dLG$}, if $\cLG(\eqx)\in\CX$ for all $x\in\R^d$ and
\begin{equation}\label{constancy}
\rGn\left(\cLg{\eqx}\right)=-\cLg{\eqx}\; \mbox{for all $x\in\R^d$.}
\end{equation}
\end{Definition}
Clearly, if $\rGn$ is constant on constants, then it is constant on any CAF with an appropriate image as \eqref{constancy} is always satisfied.

Conditional risk measures have been widely studied in the literature, see \cite{Follmer2011} for an overview. As already explained above the antitonicity is widely accepted as a minimal requirement for risk measures. The constant on constants property is a standard technical assumption, whereas we will only need the weaker property of constancy on an aggregation function for an CBRM. Typically conditional risk measures are also required to be monetary in the sense that they satisfy some translation invariance property which we do not require in our setting, see e.g.\ \cite{Detlefsen2005}.  Much of the literature is concerned with the study of quasiconvex or convex conditional risk measures which in our setting implies that the corresponding risk-consistent conditional systemic risk measure will satisfy risk-quasiconvexity resp.\ risk-convexity, see \thref{T1}.


After introducing all objects and properties of interest we are now able to state our decomposition theorem.
\begin{Theorem}\th\label{T1} A function $\rG: \Linf{\CF}{d}\to\Linf{\CG}{}$ is a CSRM
if and only if there exists a CAF $\dLG:\R^d\times\WR\to \R$ and a CBRM $\rGn: \Imag\cLG\to \Linf{\CG}{}$ such that $\rGn$ is constant on $\dLG$ (\thref{def:ConstOnAggr}) and 
\begin{equation}\label{decomp}
\rG \left(\eqX\right)=\rGn\left(\cLg{\eqX}\right)\quad \text{for all $\eqX\in \Linf{\CF}{d}$},
\end{equation}
where the extended CAF $\cLg{\eqX}:= \eqdLg{X(\om)}{\om}$ was introduced in \eqref{def:AssCaf}. The decomposition into $\rGn$ and $\cLG$ is unique.\\
Furthermore there is a one-to-one correspondence between additional properties of the CBRM $\rGn$ and additional risk-consistent properties of the CSRM $\rG$:
\begin{itemize}
\item $\rG$ is risk-convex iff $\rGn$ is convex;
\item $\rG$ is risk-quasiconvex iff $\rGn$ is quasiconvex;
\item $\rG$ is risk-positive homogeneous iff $\rGn$ is positive homogeneous;
\item $\rG$ is risk-regular iff $\rGn$ is constant on constants.
\end{itemize}  
Moreover, properties on constants of the CSRM $\rG$ are related to properties of the CAF $\dLG$:
\begin{itemize}
\item $\rG$ is convex on constants iff $\dLG$ is concave;
\item $\rG$ is positive homogeneous on constants iff $\dLG$ is positive homogeneous.
\end{itemize}  
\end{Theorem}

The proof of \thref{T1} is quite lengthy and needs some additional preparation and is thus postponed to Section~\ref{SEC:PROOF}. Note that it follows from the proof of \thref{T1} that the aggregation rule in \eqref{decomp} is deterministic if and only if $\rho_\CG(\R^d)\subseteq\R$.
\begin{Remark}
The decomposition \eqref{decomp} can also be established without requiring the CSRM to be risk-antitone, but to fulfill the weaker property 
\begin{equation}\label{hannes:PC} 
\drg{\eqXom}{\om}=\drg{\eqYom}{\om}\text{ a.s.}\Longrightarrow\rG(\eqX)=\rG(\eqY).
\end{equation}
Notice, however, if we only require \eqref{hannes:PC}, then the CBRM $\eta_\CG$ in \eqref{decomp} (and also $\rG$ itself, see \thref{T2} below) might not be antitone anymore.
\end{Remark}
An important question is to which degree CSRM's fulfill the usual (conditional) axioms of risk measures on $\Linf{\CF}{d}$ (where these axioms on $\Linf{\CF}{d}$ are defined analogously to the ones on $\Linf{\CF}{}$ in \thref{def:CBRM}). In the following \thref{T2} we will investigate the relation between risk-consistent properties and properties on constants on the one side and properties of $\rG$ on $\Linf{\CF}{d}$ on the other.
\begin{Theorem}\th\label{T2}
Let $\rG$ be a CSRM. Then 
\begin{itemize}
\item risk-antitonicity together with antitonicity on constants can equivalently be replaced by antitonicity of $\rG$ ($\eqX\geq \eqY$ implies $\rG(\eqX)\leq \rG(\eqY)$) together with \eqref{hannes:PC}.
\end{itemize}
Moreover:
\begin{itemize}
\item $\rG$ is risk-positive homogeneous and positive homogeneous on constants iff $\rG$ is positive homogeneous;
\item If $\rG$ is risk-convex and convex on constants, then $\rG$ is convex;
\item If $\rG$ is risk-quasiconvex and convex on constants, then $\rG$ is quasiconvex.
\end{itemize} 
\end{Theorem}
As for \thref{T1} we postpone the proof to Section~\ref{SEC:PROOF}.
\begin{Remark}
We have seen in \thref{T2} that a property on $\Linf{\CF}{d}$ of a CSRM is implied by the corresponding risk-consistent property and the property on constants. The reverse is only true for the antitonicity and positive homogeneity. To see this we give a counterexample for the convex case. Suppose that $\dLG(x):=u^{-1}\left(\sum_{i=1}^d x_i\right)$ and $\rGn(\eqF):=-u^{-1}\left(\BEW{\PW}{u(\eqF)}{\CG}\right)$, where $u:\R\to\R$ is a strictly increasing and convex function. Then it can be easily verified that $u^{-1}$ is strictly increasing and concave. Hence $\dLG$ is a concave CAF and $\rGn$ is a CBRM. Nevertheless, there are functions $u$ such that $\rGn$ is not a convex CBRM, e.g.\ $u(c)=c\ind_{\{c\leq0\}}+a c\ind_{\{c>0\}},a>1$.
According to \thref{T1} we get a CSRM $\rG$ by composing $\cLG$ and $\rGn$, which is explicitly given by
$$\rG(\eqX)=-u^{-1}\left(\BEW{\PW}{\sum_{i=1}^d \eqX_i}{\CG}\right).$$
It is obvious that $\rG$ is convex. But since $\rGn$ is not convex, $\rG$ cannot be risk-convex by \thref{T1}.
\end{Remark}

\section{Proof of Theorem~\ref{T1} and \ref{T2}}\label{SEC:PROOF}
Before we state the proofs of Theorems~\ref{T1} and \ref{T2}, we provide some auxiliary results.
\begin{Lemma}\th\label{l1}
Let $\dLG:\R^d\times \WR\to \R$ be a CAF and let $\CH$ be a sub-$\sigma$-algebra of $\CF$ such that $\CG\subseteq\CH\subseteq\CF$. Then 
\begin{equation}\label{reich}
\cLg{\Linf{\CH}{d}}\subseteq\Linf{\CH}{},
\end{equation}
and for every $\eqX,\eqY\in\Linf{\CH}{}$ and $\eqalpha\in\Linf{\CG}{}$ with $0\leq\alpha\leq1$ there is an $\eqF\in\Linf{\CH}{}$ such that 
$$\eqalpha\cLG(\eqX)+(\eqone-\eqalpha)\cLG(\eqY)= \cLG(\eqF\eqvecone).$$
In particular this implies that the image of $\cLG$ is $\CG$-conditionally convex.\\
Conversely, we have that
$$\Linf{\CH}{}\cap \Imag\cLG\subseteq \cLg{\Linf{\CH}{d}}.$$
\end{Lemma}
\begin{proof}
Let $\eqX\in\Linf{\CH}{d}$ and set $F(\om):=\dLg{X(\om)}{\om}$, $\omega\in\Omega$.
Since $\dLG$ is a Carath\'eodory map it follows that $F$ is $\CH$-measurable, cf.\ Lemma 8.2.3 in \cite{Aubin2009}. Let $A:=\{\om\in\WR:~\dLg{X(\om)}{\om}\leq0\}$. Then 
\begin{align}
\supnorm{F}&=\supnorm{\dLg{X(\cdot)}{\cdot}}\leq\supnorm{\dLg{\essinf X}{\cdot}\ind_A}+\supnorm{\dLg{\esssup X}{\cdot}\ind_{A^C}}\nonumber \\
&\leq\supnorm{\dLg{\essinf X}{\cdot}}+\supnorm{\dLg{\esssup X}{\cdot}}<\infty,\label{eq:gregor:1}
\end{align}
where we used the boundedness condition \thref{def:agg}~(i) in the last step and where $\essinf X:=(\essinf X_1,\ldots, \essinf X_d)$, and similarly for $\esssup$.
Hence, we conclude that $\eqF\in\Linf{\CH}{}$.\\
Let $\eqX,\eqY\in\Linf{\CH}{d}$ and $\eqalpha\in\Linf{\CG}{}$ with $0\leq\alpha\leq1$. 
The rest of the proof is based on a measurable selection theorem for which we need that the probability space is complete. 
However, $\Linf{\WR,\CH,\PW}{d}$ and $\Linf{\WR,\widehat\CH,\widehat\PW}{d}$ are isometric isomorph, where $(\WR,\widehat\CH,\widehat\PW)$ denotes the completion of $(\WR,\CH,\PW)$. 
Thus for $\eqX$ and $\eqY$ there exist respective $\eqHX,\eqHY\in\Linf{\widehat\CH}{d}$ and it is easily verified that any representatives of the equivalence classes $\eqHX$ $(\eqHY)$ and $\eqX$ $(\eqY)$ only differ on a $\widehat\PW$-nullset. 
Define 
$$\underline{x}:=\essinf\left(\min_{i=1,...,d}\left(\min(\widehat{X}_i,\widehat{Y}_i)\right)\right)\quad\text{and}\quad\xo:=\esssup\left(\max_{i=1,...,d}\left(\max(\widehat{X}_i,\widehat{Y}_i)\right)\right).$$
Since both $\widehat{X},\widehat{Y}$ are essentially bounded we have that $\underline{x},\xo\in\R$.
Moreover the random variable $G$ which is given for each $\om\in\WR$ by
$$G(\om):=\alpha(\om)\dLg{X(\om)}{\om}+\big(1-\alpha(\om)\big)\dLg{Y(\om)}{\om},$$
is contained in an equivalence class in $\Linf{\CH}{}$ by the first part of the proof and thus we can find a corresponding equivalence class $\eq{\widehat{G}}\in\Linf{\widehat\CH}{}$. 
By isotonicity we have
$$\dLg{\underline{x}\vecone}{\om}\leq \widehat{G}(\om)\leq\dLg{\xo\vecone}{\om}\quad \widehat\PW\text{-}a.s.$$
The continuity of the function $\R\ni x\mapsto \dLg{x\vecone}{\om}$ for each $\om\in\WR$ implies that
$$\widehat{G}(\om)\in\left\{\dLg{x\vecone}{\om}:x\in[\underline{x},\xo]\right\}\quad \widehat{\PW}\text{-}a.s.$$
Finally, we can apply Filippov's theorem (see e.g.\ \cite{Aubin2009} Theorem 8.2.10), that is there exists a $\widehat{\CH}$-measurable selection $\widehat{F}(\om)\in[\underline{x},\xo]$ such that 
$$\widehat{G}(\om)=\dLg{\widehat{F}(\om)\vecone}{\om}\quad \widehat{\PW}\text{-}a.s.$$
For this measurable selection $\eq{\widehat{F}}$ we can find an $\eqF\in\Linf{\CH}{}$ such that $\widehat\PW(\widehat{F}\neq F)=0$. Hence there exists an $\eqF\in\Linf{\CH}{}$ such that
$$\eqalpha\cLG(\eqX)+(\eqone-\eqalpha)\cLG(\eqY)= \cLG(\eqF\eqvecone).$$

For the last part of the proof let $\eqG\in\Imag\cLG\cap\Linf{\CH}{}$, then by definition there exists an $\eqX\in\Linf{\CF}{d}$ such that $\cLG(\eqX)=\eqG$. 
Thus by setting $\underline{x}:=\essinf(\min_{i=1,...,d}X_i)$ and $\xo:=\esssup(\max_{i=1,...,d}X_i)$ we have that
$$\dLg{\underline{x}\vecone}{\om}\leq G(\om)\leq\dLg{\xo\vecone}{\om} \quad a.s.$$
Moreover, since $G$ is $\CH$-measurable, we obtain by a similar argumentation as above that there exists a $\CH$-measurable $F$ with $\underline{x}\leq F\leq\xo$ and $\cLG(\eqF\eqvecone)=\eqG$.
\end{proof}

\begin{Lemma}\th\label{LgAC}
Let $\dLG$ be a conditional aggregation function.
Then there exists a $\PW$-nullset $N$ such that if $x,y\in\R^d$ satisfy $\dLg{x}{\om}=\dLg{y}{\om}\text{ a.s.}$ it holds that $\dLg{x}{\om}=\dLg{y}{\om}\text{ for all }\om\in N^C,$ where $N^C$ denotes the complement of $N$.
\end{Lemma}
\begin{proof}
Consider the sets $B:=\{(x,y)\in\Q^{2d}:~\dLg{x}{\om}\geq \dLg{y}{\om}\text{ a.s.}\}$ and $N_{(x,y)}:=\{\om\in\WR:~\dLg{x}{\om}<\dLg{y}{\om}\}$ for $(x,y)\in B.$
By definition $N_{(x,y)}$ is a $\PW$-nullset for all $(x,y)\in B$, but since $B$ has only countable many elements, the same holds true for the union $N:=\bigcup_{(x,y)\in B}N_{(x,y)}$.\\
Now consider $x,y\in\R^d$ such that $\dLg{x}{\om}\geq \dLg{y}{\om}\text{ a.s.}$
We can always find sequences $(x_n)_{n\in\N},(y_n)_{n\in\N}\in\Q^{\N}$ such that $x_n\downarrow x$ and $y_n\uparrow y$ for $n\to\infty$.
The isotonicity of $\dLG$ yields $\dLg{x_n}{\om}\geq\dLg{x}{\om}\geq\dLg{y}{\om}\geq\dLg{y_n}{\om}$ a.s., thus $(x_n,y_n)\in B$ for all $n\in\N$.
Therefore we get for all $\om\in N^C$ that
$$\dLg{x}{\om}=\lim_{n\to\infty}\dLg{x_n}{\om}\geq\lim_{n\to\infty}\dLg{y_n}{\om}=\dLg{y}{\om},$$
where we have used that $\dLg{\cdot}{\om}$ is continuous for every $\om\in \WR$.
As $\dLg{x}{\om}=\dLg{y}{\om}$ a.s.\ implies $\dLg{x}{\om}\geq \dLg{y}{\om}$ a.s.\ and $\dLg{x}{\om}\leq \dLg{y}{\om}$ a.s., the assertion follows.
\end{proof}
Note that the $\PW$-nullset $N$ in \thref{LgAC} is universal in the sense that it does not depend on the pair $(x,y)\in\R^{2d}$.

\begin{proof}[Proof of \thref{T1}]
For the rest of the proof let $\eqX,\eqY\in \Linf{\CF}{d}$. \\
\underline{"$\Leftarrow$"}:\\
Suppose that $\dLG:\R^d\times \WR\to \R$ is a CAF with extended CAF $\cLG:\Linf{\CF}{d}\to \Linf{\CF}{}$, and that $\rGn: \Imag\cLG\to \Linf{\CG}{}$ is a CBRM which is constant on $\dLG$.
Moreover, define the function 
$$\rG: \Linf{\CF}{d}\to\Linf{\CG}{},~\eqX\mapsto\rGn\left(\cLg{\eqX}\right).$$		
First we will show that $\rG$ is antitone (and thus in particular antitone on constants): To this end, let $\eqX\geq\eqY$.
As $\dLg{\cdot}{\om}$ is isotone for all $\om\in\WR$ we know from \eqref{hannes:CFisotone} that also the extended CAF is isotone, i.e. $\cLg{\eqX}\geq\cLg{\eqY}$. By the antitonicity of $\rGn$ we can conclude that
$$\rG\left(\eqX\right)=\rGn\left(\cLg{\eqX}\right)\leq\rGn\left(\cLg{\eqY}\right)=\rG\left(\eqY\right).$$
Next we will show that there exists a realization of $\rG$ with continuous paths and which fulfills the risk-antitonicity.
From \eqref{constancy} and \thref{LgAC} it can be readily seen that we can always find a realization of $\eta_\CG$
and a universal $\PW$-nullset $N$  such that for all $\om\in N^C$
\begin{equation}\label{gregor:eq:4}\rgn{\cLg{\eqx}}{\om}=-\dLg{x}{\om}\text{ for all }x\in\R^d.\end{equation}
Given this realization of $\rGn$ we consider in the following the realization $\rG(\cdot,\cdot)$ of $\rG$ given by
$$\crg{\eqX}{\om}:=\rgn{\cLg{\eqX}}{\om},\; \eqX\in\Linf{\CF}{d}, \om\in\WR.$$ 
The function $\drG:\R^d\times\WR\to\R;~x\mapsto\rG(\eqx,\om)$ has continuous paths (a.s.) because $\dLG$ has continuous paths. 
As for the risk-antitonicity, let $\drg{\eqXom}{\om}\geq\drg{\eqYom}{\om}$ a.s.
By rewriting this in terms of the decomposition, i.e.\ \\ $\rGn\big(\cLg{\eq\eqXom},\om\big)\geq\rGn\big(\cLg{\eq\eqYom},\om\big),$ we realize by \eqref{gregor:eq:4} that
\begin{equation}\label{gregor:eq:5}\dLg{X(\om)}{\om}\leq\dLg{Y(\om)}{\om}\text{ a.s.}\end{equation}
Note that our application of \eqref{gregor:eq:4} relies on the fact that the nullset $N$ in \eqref{gregor:eq:4} does not depend on $x\in \R^d$. As \eqref{gregor:eq:5} is equivalent to $\cLg{\eqX}\leq\cLg{\eqY}$, we conclude that
$$\rG(\eqX)=\rGn(\cLg{\eqX})\geq\rGn(\cLg{\eqY})=\rG(\eqY),$$
where we used the antitonicity of $\rGn$.
Hence, we have proved that $\rG$ is a CSRM.

\smallskip\noindent 
Next we treat the special cases when $\rGn$ and/or $\dLG$ satisfy some extra properties. \\
\emph{Risk-regularity}: Suppose $\rGn$ is constant on constants. Then we have
$$\rG(\eqX)=-\cLg{\eqX}\text{ for all }\eqX\in\Linf{\CG}{d},$$
and thus we obtain for the realization $\crg{\cdot}{\cdot}$ that for all $\eqX\in\Linf{\CG}{d}$
$$\crg{\eqX}{\om}=-\dLg{X(\om)}{\om}\text{ a.s.}$$
As above \eqref{gregor:eq:4} implies that for all $\om \in N^C$  
$$-\dLg{X(\om)}{\om} = \rgn{\cLg{\eq\eqXom}}{\om}=\drg{\eqXom}{\om}.$$
\emph{Risk-quasiconvexity/convexity}: Suppose that $\rGn$ is quasiconvex. We show that $\rG$ is risk-quasiconvex. To this end, suppose there exist $\eqX,\eqY,\eqZ\in\Linf{\CF}{d}$ and an $\eqalpha\in\Linf{\CG}{}$ with $0\leq\alpha\leq1$ such that
$$\drg{\eqZom}{\om}=\alpha(\om)\drg{\eqXom}{\om}+\big(1-\alpha(\om)\big)\drg{\eqYom}{\om}\text{ a.s.}$$ 	
Then, as above, by using \eqref{gregor:eq:4}, it follows that 
$$\cLg{\eqZ}=\eqalpha\cLg{\eqX}+(\eqone-\eqalpha)\cLg{\eqY}.$$
Hence the quasiconvexity of $\rGn$ yields 
\begin{align*}
\rG(\eqZ)	&=\rGn\left(\cLg{\eqZ}\right)=\rGn\left(\eqalpha\cLg{\eqX}+(\eqone-\eqalpha)\cLg{\eqY}\right)\\
&\leq\rGn\left(\cLg{\eqX}\right)\vee\rGn\left(\cLg{\eqY}\right)\\
&=\rG(\eqX)\vee\rG(\eqY).
\end{align*}
Similarly it follows that $\rho_\CG$ is risk-convex whenever $\eta_\CG$ is convex.\\
\emph{Risk-positive homogeneity}: Finally, if $\rGn$ is positively homogeneous, then it is straightforward to see that also $\rG$ is risk-positively homogeneous.\\
\emph{Properties on constants:} Suppose that $\dLG$ is concave or positive homogeneous, then it is an immediate consequence of \eqref{gregor:eq:4} that $\rG$ is convex on constants or positive homogeneous on constants, resp.

	
\noindent\underline{"$\Rightarrow$"}:\\
Let $\crg{\cdot}{\cdot}$ denote a realization of the CSRM $\rG$ such that $\drG$ has continuous paths and the risk-antitonicity holds. We define the function $\HLG:\R^d\times\WR\to \R$ by
	\begin{equation}\label{mm0}
		\HLg{x}{\om}:=-\drg{x}{\om}.
	\end{equation}
We show that $\HLG(\cdot,\om)$ is a DAF for almost all $\om\in\WR$, i.e.\ that it is isotone and continuous. The continuity is obvious by $\eqref{mm0}$. For the isotonicity consider the sets $B:=\left\{(x,y)\in\Q^{2d}:~x\geq y\right\}$ and $A^{(1)}_{(x,y)}:=\{\om\in\WR:~\HLg{x}{\om}<\HLg{y}{\om}\}$ for $(x,y)\in B.$
Since $\rG$ is antitone on constants we obtain that $A^{(1)}:=\bigcup_{(x,y)\in B}A^{(1)}_{(x,y)}$ is a $\PW$-nullset.
Moreover, let $A^{(2)}$ denote the $\PW$-nullset on which $\HLG$ has discontinuous sample paths. Consider $x,y\in\R^d$ such that $x\geq y$, and let $(x_n,y_n)\in B^\N$ be a sequence which converges to $(x,y)$ for $n\to\infty$. Then we get for all $\om\in\left(A^{(1)}\cup A^{(2)}\right)^C$ that
$$
\HLg{x}{\om}=\lim_{n\to\infty}\HLg{x_n}{\om}\geq\lim_{n\to\infty}\HLg{y_n}{\om}=\HLg{y}{\om},
$$ 
and thus the paths $\HLg{\cdot}{\om}$ are isotone a.s. 

The fact that the paths $\HLg{\cdot}{\om}$ are concave (positively homogeneous) a.s.\ whenever $\rG$ is convex on constants (positively homogeneous on constants) follows by a similar approximation argument on the continuous paths which are concave (positively homogeneous) on $\Q^d$.  

Given the above considerations, we choose a modification $\dLG$ of $\HLG$ such that  $\dLG(\cdot, \om)$, is a (concave/positively homogeneous) DAF for all $\om\in\WR$. Note that for $\dLG$ relation \eqref{mm0} is only valid a.s., that is there is a $\PW$-nullset $N$ such that for all $x\in\R^d$ and $\om\in N^C$
\begin{equation}\label{mm1}
\dLg{x}{\om}=-\drg{x}{\om}.
\end{equation}
As $-\drg{x}{\cdot}\in\SLinf{\CG}{}$ and thus also $\dLg{x}{\cdot}\in\SLinf{\CG}{}$ for all $x\in\R^d$ (note that $N\in\CG$), we have shown that $\dLG$ is indeed a CAF.


\smallskip\noindent
Next, we will construct a CBRM $\rGn: \Imag\cLG=:\CX\to \Linf{\CG}{}$ such that $\rho_\CG=\rGn\circ\cLG$ where $\cLG$ is the extended CAF of $\dLG$. For $\eqF\in\CX$ we define
\begin{equation}\label{mm2}
\rGn(\eqF):=\rG(\eqX),
\end{equation}
where $\eqX\in \Linf{\CF}{d}$ is given by
\begin{equation}\label{mm3}
\cLg{\eqX}=\eqF.
\end{equation}
Since $\eqF\in\CX$ the existence of such $\eqX$ is always ensured. By \eqref{mm3} and \eqref{mm2} we obtain the desired decomposition
$$\rGn\left(\cLg{\eqX}\right)=\rG(\eqX),$$
if $\rGn$ is well-defined. In order to show the latter, let $\eqXh{1},\eqXh{2}\in \Linf{\CF}{d}$ such that $$\cLg{\eqXh{1}}=\cLg{\eqXh{2}}=\eqF,$$ which by definition of $\cLG$ in \eqref{def:AssCaf} can be rewritten as
$$\dLg{\Xh{1}(\om)}{\om}=F(\om)=\dLg{\Xh{2}(\om)}{\om}\text{ a.s.}$$
By \eqref{mm1} this can be restated in terms of $\drg{\cdot}{\cdot}$ as
$$\drg{\eqXomh{1}}{\om}=\drg{\eqXomh{2}}{\om}\text{ a.s.}$$
Now the risk-antitonicity of $\rG$ yields $\rG\left(\eqXh{1}\right)=\rG\left(\eqXh{2}\right),$ so $\rGn$ in \eqref{mm2} is indeed well-defined.
	
Next we will show that $\rGn$ is a CBRM. For this purpose, let in the following $\eqF, \eqG \in\CX$ and $\eqX, \eqY \in\Linf{\CF}{d}$ be such that $\cLg{\eqX}=\eqF$, $\cLg{\eqY}=\eqG$. 

\noindent\emph{Antitonicity}: 
 Assume $\eqF\geq\eqG$. Then, by \eqref{mm1} for almost every $\om\in\WR$ 
	$$-\drg{\eqXom}{\om}=\dLg{X(\om)}{\om}=F(\om)\geq G(\om)=\dLg{Y(\om)}{\om}=-\drg{\eqYom}{\om}.$$
Hence, risk-antitonicity ensures that $\rG(\eqX)\leq \rG(\eqY)$.
But by \eqref{mm2} this is equivalent to $\rGn(\eqF)\leq\rGn(\eqG).$
	
\noindent\emph{Constancy on $\dLG$}:	Constancy on $\dLG$ is an immediate consequence of \eqref{mm1}-\eqref{mm3}, since for $x\in\R^d$
	$$\rGn\left(\cLg{\eqx}\right)=\rG(\eqx)=-\cLg{\eqx}.$$
\smallskip\noindent 	
Hence, the decomposition \eqref{decomp} is proved.\\
\noindent\emph{Uniqueness}: Let $\rGn^{(1)},\rGn^{(2)}$ be CBRM's and $\dLG^{(1)},\dLG^{(2)}$ be CAF's such that $\rGn^{(1)}$ and $\rGn^{(2)}$ are constant on $\dLG^{(1)}$ and $\dLG^{(2)}$ resp.\ and it holds that
$$\rGn^{(1)}\left(\cLG^{(1)}(\eqX)\right)=\rG(\eqX)=\rGn^{(2)}\left(\cLG^{(2)}(\eqX)\right)\quad\text{for all }\eqX\in\Linf{\CF}{d}.$$
Then it follows from the constancy on the respective CAF's that for all $x\in\R^d$ $\cLG^{(1)}(\eqx)=\cLG^{(2)}(\eqx)$, i.e.
\begin{equation}\label{eq:eqonreals}
\dLG^{(1)}(x,\om)=\dLG^{(2)}(x,\om)\quad a.s.
\end{equation}
Note that by a similar argumentation as in the proof of \thref{LgAC} \eqref{eq:eqonreals} holds true on a universal $\PW$-nullset $N$ for all $x\in\R^d$. In order to show that $\cLG^{(1)}$ and $\cLG^{(2)}$ are not only equal on constants let $\eqX\in\Linf{\CF}{d}$. Then $\eqX$ can be approximated by simple $\CF$-measurable random vectors, i.e.\ there exists a sequence $(\eqX_n)_{n\in\N}$ with $X_n\to X$ $\PW$-a.s.\ and $X_n=\sum_{i=1}^{k_n}x_i^n\ind_{A_i^n}$ for all $n\in\N$, where $x_i^n\in\R$ and $A_i^n\in\CF,i=1,...,k_n$ are disjoint sets such that $\PW(A_i^n)>0$ and $\PW\big(\bigcup_{i=1}^{k_n}A_i^n\big)=1$. Denote by $M$ the $\PW$-nullset on which $(X_n)_{n\in\N}$ does not converge. Then by the continuity property of a CAF and \eqref{eq:eqonreals} we have for all $\om\in(N\cup M)^C$ that
\begin{align*}
\dLG^{(1)}\left(X(\om),\om\right) &=\dLG^{(1)}\left(\lim_{n\to\infty}X_n(\om),\om\right)\quad 
	=\quad \lim_{n\to\infty}\dLG^{(1)}\left(X_n(\om),\om\right)\\ &=  \lim_{n\to\infty}\sum_{i=1}^{k_n}\dLG^{(1)}\left(x_i^n,\om\right)\ind_{A_i^n}(\om)\quad 
	=\quad \lim_{n\to\infty}\sum_{i=1}^{k_n}\dLG^{(2)}\left(x_i^n,\om\right)\ind_{A_i^n}(\om)\\
	&=\dLG^{(2)}\left(X(\om),\om\right),
\end{align*}
and thus $\cLG^{(1)}(\eqX)=\cLG^{(2)}(\eqX)$ for all $\eqX\in\Linf{\CF}{d}$.
Finally for all $\eqF\in\Imag\cLG^{(1)}=\Imag\cLG^{(2)}$ there is an $\eqX\in\Linf{\CF}{d}$ such that $\cLG^{(1)}(\eqX)=\cLG^{(2)}(\eqX)=\eqF$ and hence
$$\rGn^{(1)}(\eqF)=\rG(\eqX)=\rGn^{(2)}(\eqF).$$
Next we consider the cases when $\rG$ fulfills some additional properties. \\
\emph{Constant on constants}: Let $\rG$ be risk-regular. Then \eqref{mm1} implies that for  all $\eqX\in\Linf{\CG}{d}$
$$\crg{\eqX}{\om}=\drg{\eqXom}{\om}=-\dLg{X(\om)}{\om}\text{ a.s.}$$
and hence $\rG(\eqX)=-\cLg{\eqX}$.
Let now $\eqF\in \CX\cap\Linf{\CG}{}$. By the definition of $\CX$ and \thref{l1} we know that there exists a $\eqX\in\Linf{\CG}{d}$ such that $\cLG\left(\eqX\right)=\eqF$.
We thus obtain by \eqref{mm2} that 
$$\rGn(\eqF)=\rG(\eqX)=-\cLg{\eqX}=-\eqF.$$

\noindent\emph{Quasiconvexity/convexity}: Let $\rG$ be risk-quasiconvex. 
Let $\eqalpha\in\Linf{\CG}{}$ with $0\leq\alpha\leq1$ and set $\eqH:=\eqalpha \eqF+(\eqone-\eqalpha)\eqG,$
where $\eqF,\eqG\in\CX$, and $\eqX, \eqY\in\Linf{\CF}{d}$ are such that $\cLg{\eqX}=\eqF$, $\cLg{\eqY}=\eqG$. Note that since $\CX$ is $\CG$-conditionally convex, $\eqH\in\CX$ and thus there exists a $\eqZ\in\Linf{\CF}{d}$ with $\cLg{\eqZ}=\eqH$.  Then 
\begin{align*}
\dLg{Z(\om)}{\om}=H(\om)&=\alpha(\om) F(\om)+(1-\alpha(\om))G(\om)\\
&=\alpha(\om) \dLg{X(\om)}{\om}+(1-\alpha(\om))\dLg{Y(\om)}{\om}\text{ a.s.}
\end{align*}
Thus it follows by \eqref{mm1} 
$$\drg{\eqZom}{\om}=\alpha(\om) \drg{\eqXom}{\om}+(1-\alpha(\om))\drg{\eqYom}{\om}\text{ a.s.},$$
which in conjunction with the risk-quasiconvexity of $\rG$ results in
$$\rGn(\eqH)=\rG(\eqZ)\leq\rG(\eqX)\vee\rG(\eqY)=\rGn(\eqF)\vee\rGn(\eqG).$$

Similarly one shows that $\rGn$ is convex if $\rG$ is risk-convex.

\noindent\emph{Positive homogeneity}: Let $\rG$ be risk-positively homogeneous. Further let $\eqF\in\CX$, $\eqX\in \Linf{\CF}{d}$ with $\cLg{\eqX}=\eqF$, and let $\eqalpha\in\Linf{\CG}{}$ with $\alpha\geq0$ and $\eqalpha \eqF=:\eqG\in\CX$. Then there is also a $\eqY\in \Linf{\CF}{d}$ with $\cLg{\eqY}=\eqG$. Moreover, $\dLg{Y(\om)}{\om}=\alpha(\om)\dLg{X(\om)}{\om}$ a.s. Hence, by \eqref{mm1} in conjunction with the risk-positive homogeneity we obtain that $\rG(\eqY)=\eqalpha\rG(\eqX)$. Consequently,
$$\rGn(\eqalpha \eqF)=\rG(\eqY)=\eqalpha \rG(\eq{X})=\eqalpha\rGn(\eqF).$$
\end{proof}
\begin{proof}[Proof of \thref{T2}]
As $\rG$ is risk-antitone and antitone on constants, it is obvious that $\rG$ also fulfills \eqref{hannes:PC}. 
Furthermore, we already showed, based on the antitonicity on constants and continuous paths requirements, in the proof of \thref{T1} that $\omega\mapsto \drg{x}{\om} (=-\widehat \Lambda_\CG(x,\omega))$ has almost surely antitone paths. 
Hence, we have for all $\eqX,\eqY\in\Linf{\CF}{d}$ with $\eqX\geq \eqY$, that $\drg{\eqXom}{\om}\leq\drg{\eqYom}{\om}\text{ a.s.}$ and thus the risk-antitonicity yields $\rG(\eqX)\leq\rG(\eqY)$.
Hence we conclude that $\rG$ is antitone.\\
For the converse implication let $\rG$ be antitone and let $\crg{\cdot}{\cdot}$ be a realization with corresponding restriction $\drg{\cdot}{\cdot}$ which fulfills \eqref{hannes:PC}. 
The antitonicity on constants is an immediate consequence of the much stronger antitonicity on $L^\infty$ of $\rG$. 
By reconsidering the proof of \thref{T1}, we observe that we may replace the risk-antitonicity by \eqref{hannes:PC} when extracting the aggregation function. 
Hence, \eqref{hannes:PC} is sufficient to construct a modification of $\crg{\cdot}{\cdot}$ and thus of $\drg{\cdot}{\cdot}$ such that $\drg{\cdot}{\cdot}$ has surely continuous and antitone paths. 
Therefore, suppose that $\crg{\cdot}{\cdot}$ is already this realization. 
Now let $\eqX,\eqY\in\Linf{\CF}{d}$ with $\drg{\eqXom}{\om}\leq\drg{\eqYom}{\om}\text{ a.s.}$
According to \thref{l1} with $\dLG$ as in \eqref{mm1} there are $\eqF,\eqG\in\Linf{\CF}{}$ such that
$$\drg{\eqFom\vecone}{\om}=\drg{\eqXom}{\om}\leq\drg{\eqYom}{\om}=\drg{\eqGom\vecone}{\om}\text{ a.s.}$$
As the paths of $\drg{\cdot}{\cdot}$ are antitone, it can be readily seen that $\eqF\geq\eqG$ on $A:=\left\{\om\in\WR:\drg{\eqXom}{\om}<\drg{\eqYom}{\om}\right\}$.
Now set $\eqH:=\eqG\eqind_A+\eqF\eqind_{A^C}\in\Linf{\CF}{}$. Then $\eqF\geq\eqH$ and $\drg{\eqYom}{\om}=\drg{\eqHom\vecone}{\om}$ a.s.
Hence it follows from \eqref{hannes:PC} and the antitonicity of $\rG$ that
$$\rG(\eqX)=\rG(\eqF\eqvecone)\leq\rG(\eqH\eqvecone)=\rG(\eqY).$$
This completes the proof of the first equivalence in \thref{T2}.

Let $\rG$ be risk-positive homogeneous and positive homogeneous on constants.
Since all requirements of \thref{T1} are met, we also have that $\R^d\ni x\mapsto\drg{x}{\om}$ is almost surely positive homogeneous.
Therefore we obtain for all $\eqX\in\Linf{\CF}{d}$ and $\eqalpha\in\Linf{\CG}{}$ with $\alpha\geq0$ that 
$$\drg{\alpha(\om) \eqXom}{\om}=\alpha(\om)\drg{\eqXom}{\om}\text{ a.s.},$$
and hence the risk-positive homogeneity implies $\rG(\eqalpha \eqX)=\eqalpha\rG(\eqX)$ which is positive homogeneity of $\rG$.\\
Conversely, if $\rG$ is positive homogeneous, then it is also positive homogeneous on constants as well as for almost all paths of the realization.
Hence, if there exists $\eqX,\eqZ\in\Linf{\CF}{d}$ and $\eqalpha\in\Linf{\CG}{}$ with $\alpha\geq0$ such that 
$$\drg{\eqZom}{\om}=\alpha(\om)\drg{\eqXom}{\om}\text{ a.s.},$$
then the right-hand-side equals $\drg{\alpha(\om) \eqXom}{\om}$ a.s.
Using \eqref{hannes:PC} and the positive homogeneity of $\rG$ we conclude that
$$\rG(\eqZ)=\rG(\eqalpha \eqX)=\eqalpha\rG(\eqX).$$

Let $\rG$ be risk-convex and convex on constants. First we will show that risk-convexity is equivalent to the following property:  
If for $\eqX,\eqY,\eqZ\in\Linf{\CF}{d}$ there exists a $\eqalpha\in\Linf{\CG}{}$ with $0\leq\alpha\leq 1$ such that
\begin{align}
	&\label{eq:T2:1}\drg{\eqZom}{\om}\leq\alpha(\om)\drg{\eqXom}{\om}+\big(1-\alpha(\om)\big)\drg{\eqYom}{\om}\text{ a.s.},\\
	& \text{then}\;   \rG(\eqZ)\leq\eqalpha\rG(\eqX)+\big(\eqone-\eqalpha\big)\rG(\eqY). \nonumber 
\end{align}
On the one hand, it is obvious that \eqref{eq:T2:1} implies risk-convexity.
On the other hand, let $\eqZh{1}\in\Linf{\CF}{d}$ such that 
\begin{equation*}\label{eq:T2:2}
\drg{\eqZhom{1}}{\om}\leq\alpha(\om)\drg{\eqXom}{\om}+\big(1-\alpha(\om)\big)\drg{\eqYom}{\om}\text{ a.s.}
\end{equation*}
We know by \thref{l1} that there is a $\eqZh{2}\in\Linf{\CF}{d}$ such that 
$$\alpha(\om)\drg{\eqXom}{\om}+\big(1-\alpha(\om)\big)\drg{\eqYom}{\om}=\drg{\eqZhom{2}}{\om}\text{ a.s.}$$
By the risk-convexity we obtain that 
$$\rG(\eqZh{2})\leq\eqalpha\rG(\eqX)+(\eqone-\eqalpha)\rG(\eqY).$$
As risk-antitonicity implies $\rG(\eqZh{1})\leq\rG(\eqZh{2})$, we 
conclude that risk-convexity and \eqref{eq:T2:1} are equivalent. 
Next we show the convexity of $\rG$. 
To this end let $\eqX,\eqY\in\Linf{\CF}{d}$ and $\eqalpha\in\Linf{\CG}{}$ with $0\leq\alpha\leq 1$.
Once again we can reason as in the proof of \thref{T1} that $\R^d\ni x\mapsto\drg{\eqx}{\om}$ is almost surely convex, because $\rG$ has continuous paths and is convex on constants.
Thus we have that
$$\drg{(\alpha X+(1-\alpha)Y)(\om)}{\om}\leq\alpha(\om)\drg{\eqXom}{\om}+\big(1-\alpha(\om)\big)\drg{\eqYom}{\om}\text{ a.s.}$$
Now \eqref{eq:T2:1} implies that
$$\rG(\eqalpha \eqX+(\eqone-\eqalpha)\eqY)\leq\eqalpha\rG(\eqX)+(\eqone-\eqalpha)\rG(\eqY),$$
which is the desired convexity of $\rG$.

The other assertion concerning risk-quasiconvexity follows in a similar way.
\end{proof}
\section{Examples}\label{SEC:EX}
\begin{Example}
	As already mentioned in the introduction a typical aggregation function when dealing with multidimensional risks is $$\dL_\text{sum}(x)=\sum_{i=1}^d x_i,~x\in\R^d.$$
	However, such an aggregation rule might not always be reasonable when measuring systemic risk. The main reason for this is the limited transferability of profits and losses between institutions of a financial system. An alternative popular aggregation function which does not allow for a subsidization of losses by other profitable institutions is given by
	$$\dL_{\text{loss}}(x)=\sum_{i=1}^d-x^-_i,~x\in\R^d,$$
	where $x^-_i=-\min\{x_i,0\}$; see \thref{ex:DIP}. Obviously, both $\dL_\text{sum}$ and $\dL_\text{loss}$ are DAF's which are additionally concave and positive homogeneous.
\end{Example}
\begin{Example}[Countercyclical regulation]\th\label{ex:Countercyclical}
Risk charges based on systemic risk measures typically will increase drastically in a distressed market situation which might even worsen the crisis further. Therefore one might argue that, for instance in a recession where also the real economy is affected, the financial regulation should be relaxed in order to stabilize the real economy, cf.\ \cite{Brunnermeier2013}.
In our setup we can incorporate such a dynamic countercyclical regulation as follows:

Let $(\WR,\CF,(\CF_t)_{t\in\{0,...,T\}},\PW)$ be a filtered probability space, where $\CF_T=\CF$. Let $(x,y)\in\R^{2d}$ be the profits/losses of the financial system, where the first $d$ components $x$ are the profits/losses from contractual obligations with the real economy and $y$ are the profits/losses from other obligations. Moreover let $Y(t)$, $t=-1,0,...,T-1$ be the gross domestic product (GDP) process with $Y(t)\in\SLinf{\CF_t}{}$, $t=0,...,T-1$, and $Y(-1)\in \R^+\backslash\{0\}$.
Suppose that the regulator sees the economy in distress at time $t$, if the GDP process $Y(t)$ is less than $(1+\theta)Y(t-1)$ for some $\theta\in\R$. We assume that in those scenarios the regulator is interested to lower the regulation in order to give incentives to the financial system for the supply of additional credit to the real economy.
This policy might lead to the following dynamic conditional aggregation function from the perspective of the regulator
$$\dL\big((x,y),t,\om):=-\sum_{i=1}^d\left(\alpha\ind_{A(t)}(\om)+\ind_{A(t)^C}(\om)\right) x_i^-+y_i^-, t=0,...,T-1,$$
where $\alpha\in[0,1)$ and $A(t)=\left\{Y(t)\leq (1+\theta)Y(t-1)\right\}$ for $t=0,...,T-1$. Obviously, $\dL\big((x,y),t,\om)$ is a CAF with respect to $\CF_t$ which is positive homogeneous and concave.
\end{Example}

\begin{Example}[Too big to fail]\th\label{tbtf}
In this example we will consider a dynamic conditional aggregation function which depends on the relative size of the interbank liabilities. For instance, \cite{Cont2013} find that for the Brasilian banking network there is a strong connection between the size of the interbank liabilities of a financial institution and its systemic importance. This fact is often quoted as 'too big to fail'. 

Let $(\WR,\CF,(\CF_t)_{t\in\{0,...,T\}},\PW)$ be a filtered probability space, where $\CF_T=\CF$. Moreover, let $L_i(t)\in\SLinf{\CF_t}{}$ denote the sum of all liabilities at time $t$ of institution $i\in \{1,\ldots, d\}$ to any other banks. Then 
$$\alpha_i(t):=\frac{L_i(t)}{\sum_{j=1}^d L_j(t)},\quad t=0,...,T-1,$$
is the relative size of its interbank liabilities.
Now consider the following conditional extension of an aggregation function which was proposed in \cite{Brunnermeier2013}:
$$\dL_{\text{BC}}(x,t,\om)=\sum_{i=1}^d -\alpha_i(t,\om) x_i^-+\beta_i(\theta_i-x_i)^-,\quad t=0,...,T-1,$$
where $\beta,\theta\geq0$. Firstly, this conditional aggregation function always takes losses into consideration, whereas profits of a financial institution $i$ are only accounted for if they are above a firm specific threshold $\theta_i$. Secondly, profits are weighted by the deterministic factor $\beta$ and the losses are weighted proportional to the liability size of the corresponding financial institution at time $t$. Therefore losses from large institutions, which are more likely to be systemically relevant, contribute more to the total risk. \\
$\dL_{\text{BC}}(\cdot,t,\cdot)$ is a CAF which, however, in general is neither quasiconvex nor positively homogeneous as it may be partly flat depending on $\theta$. 
\end{Example}

\begin{Example}\th\label{ex:preference}
Suppose that the regulator of the financial system has certain preferences on the distribution of the total loss amongst the financial institutions. For instance he might prefer a situation when a number of financial institutions face a relatively small loss each in front of a situation in which one financial institution experiences a relatively large loss. 
Such a preference can be incorporated by the following aggregation function 
$$\dL_{\text{exp}}(x)=\sum_{i=1}^d -x^-_i\ind_{\{x_i>\theta_i\}}+\left(\frac{1}{\gamma_i}\left(1-e^{\gamma_i \left(x_i^-+\theta_i\right)}\right)+\theta_i\right)\ind_{\{x_i\leq \theta_i\}},$$
where $\theta_i\leq0$ and $\gamma_i>0$ for $i=1,...,d$.
That is, if the losses of firm $i$ exceed a certain threshold $\theta_i$, e.g.\ a certain percentage of the equity value, then the losses are accounted for exponentially.
\end{Example}

\begin{Example}[Stochastic discount]\th\label{EX:Stochasticdiscount}
Suppose that $D\in\SLinf{\CF}{}$ is some $\CG$-measurable stochastic discount factor. 
A typical approach to define monetary risk measurement of some future risk is to consider the discounted risks. Consider any (conditional) aggregation function $\dL$, which does not discount in aggregation, such as $\dL_{\text{sum}}$, $\dL_{\text{loss}}$, or $\dL_{\text{BC}}$, etc. Then the discounted monetary aggregated risk is $D\dL(X)$.  If $\dL$ is positively homogeneous, then $D\dL(X)=\dL(DX)$ which is the aggregated risk of the discounted system $DX$. However, if $\dL$ is not positively homogeneous - such as $\dL_{\text{BC}}$ or $\dL_{\text{exp}}$ - then the discounted aggregated risk can only be formulated in terms of the conditional aggregation function $$\dLg{x}{\om}:=\dL(x)D(\om).$$
\end{Example}

\begin{Example}[CoVaR]\th\label{ex:COVAR}
In this example we will consider the CoVaR proposed in \cite{Adrian2011}; see \eqref{eq:covar:intro}. To this end, we first recall the (conditional) Value at Risk: We denote the Value at Risk at level $q\in(0,1)$ by
$$\VaR_q(\eqF)=-\inf_{x\in\R}\{\PW(F\leq x)>q\}.$$
Furthermore, the conditional VaR at level $q\in(0,1)$ is defined as
$$\VaR_q(\eqF|\CG):=-\essinf_{\eqalpha\in\Linf{\CG}{}}\left\{\PW\big(F\leq\alpha\;\big|\;\CG\big)> q\right\},$$
c.f. \cite{Follmer2011}.
The conditional VaR is positive homogeneous, antitone, and constant on constants. Thus it is a CBRM which is constant on every possible CAF.
Note that, as is well-known for the unconditional case, the conditional VaR is not quasiconvex. By composing $\VaR_q(\cdot|\CG)$ with a CAF $\dLG$ we obtain a CSRM 
\begin{equation}\label{eq:gen:covar}
\rG(\eqX)=\VaR_q\left(\left.\cLG(\eqX)\right|\CG\right), \quad \eqX\in \Linf{\CF}{d},
\end{equation}
which is risk-positive homogeneous and risk-regular. 
\\
Now we consider the case where $\eqX$ represents a financial system and the CAF in \eqref{eq:gen:covar} is $\dL_{\text{sum}}$. Moreover consider the sub-$\sigma$-algebra $\CG:=\sigma(A)$ of $\CF$, where $A:=\{X_j\leq-\VaR_q(X_j)\}$ for a fixed $j\in\{1,...,d\}$.
Then the CSRM $\rG(\eqX)$ from \eqref{eq:gen:covar} evaluated in the event $A$ equals 
\begin{equation}\label{eq:covar:ex}\VaR_q\left(\left.\sum_{i=1}^d\eqX_i\right|\left\{X_j\leq-\VaR_q(X_j)\right\}\right)\end{equation}
which is the CoVaR proposed in \cite{Adrian2011}.\\
As we have already pointed out in the introduction, it is more reasonable to use an aggregation function which incorporates an explicit contagion structure. We will modify the CoVaR in this direction in \thref{EX:CM}.
\end{Example}

\begin{Example}[CoES and SES]\th\label{ex:COES}
The conditional Average Value at Risk at level $q\in(0,1)$ is given by  
$$\AVaR_q(\eqF|\CG):=\esssup_{\PWQ\in\CP^q}\BEW{\PWQ}{-\eqF}{\CG},\quad \eqF\in \Linf{\CF}{},$$
where $\CP^q$ is the set of probability measures $\PWQ$ on $(\WR,\CF)$ which are absolutely continuous w.r.t.\ $\PW$ such that $\PWQ|_\CG = \PW$ and $\frac{d\PWQ}{d\PW}\leq 1/q$ a.s. $\AVaR_q(\cdot|\CG)$ is a convex and positive homogeneous CBRM. 
Notice that the conditional Average Value at Risk can also be written as
\begin{equation}\label{eq:AVAR}\AVaR_q(\eqF|\CG)=\frac{1}{q}\BEW{\PW}{\left(\eqF+\VaR_q(\eqF|\CG)\right)^-}{\CG}+\VaR_q(\eqF|\CG),\end{equation}
cf. \cite{Follmer2011}, where $\VaR_q(\cdot|\CG)$ is discussed in \thref{ex:COVAR}.

\smallskip\noindent
As in \thref{ex:COVAR} let $\CG=\sigma\left(A\right)$ with $A=\{X_j\leq-\VaR_q(X_j)\}$ for a fixed $j\in\{1,...,d\}$ and $q\in(0,1)$.
Using \eqref{eq:AVAR}, if $$\PW(\left.F\leq -\VaR_q(\eqF|\CG)\right|\CG)=\eq{q},$$ then
\begin{align}
\AVaR_q(\eqF|\CG)	&=\BEW{\PW}{-\eqF}{\{F\leq-\VaR_q(\eqF|A)\}\cap A}\ind_A \nonumber \\
&\hspace{0.5cm}+ \BEW{\PW}{-\eqF}{\{F\leq-\VaR_q(\eqF|A^C)\}\cap A^C}\ind_{A^C}.
\label{eq:avar:ex} \end{align}
Therefore, $\rG(\eqX)=\AVaR_q(\dL_{\text{sum}}(\eqX)|\CG)$ evaluated in the event $A$ equals
$$\BEW{\PW}{-\sum_{i=1}^d \eqX_i}{\left\{\sum_{i=1}^d X_i\leq -\VaR_q\left(\sum_{i=1}^d \eqX_i~\bigg|~A\right)\right\}\cap A}.$$
In other words, $\rG(\eqX)|_A$ is the expected loss of the financial system $\eqX$ given that the loss $\eqX_j$ of institution $j$ is below $\text{VaR}_q(\eqX_j)$ and simultaneously the loss of the system is below its CoVaR $\VaR_q(\sum_{i=1}^d \eqX_i|A)$. \eqref{eq:avar:ex} corresponds to the conditional expected shortfall (CoES) proposed in \cite{Adrian2011}.

\smallskip\noindent
Now we change the point of view and consider the losses of a financial institution $\eqX_j$ given that the financial system is in distress, that is if 
$$\sum_{i=1}^d\eqX_i\leq -\VaR_q\left(\sum_{i=1}^d\eqX_i\right).$$ 
Let $\CG:=\sigma\big(\{\sum_{i=1}^dX_i\leq -\VaR_q(\sum_{i=1}^d\eqX_i)\}\big)$. By composing the DAF $\dL(x):=x_j$ and the CBRM $\rGn(\eqF)=\BEW{\PW}{-\eqF}{\CG}$ we obtain a convex and positive homogeneous CSRM
$$\rho_\CG(\eqY)=\BEW{\PW}{\eqY_j}{\CG},\quad \eqY\in\Linf{\CF}{d}.$$
$\rho_\CG(\eqX)$ evaluated on the event $\{\sum_{i=1}^dX_i\leq -\VaR_q(\sum_{i=1}^d\eqX_i)\}$ which is the so-called systemic expected shortfall ($SES^{j}$) introduced in \cite{Acharya2010}. 
\end{Example}

\begin{Example}[DIP]\th\label{ex:DIP} 
In this example we recall the distress insurance premium (DIP) proposed by \cite{Huang2012}. 
It is closely related to CoES and SES discussed in \thref{ex:COES}. 
However, instead of $\dL_{\text{sum}}$, the aggregation function is $\dL_{\text{loss}}$, that is losses cannot be subsidized by profits from the other institutions. 
The event representing the financial system in distress is $\{\cL_{\text{loss}}(X)\leq \theta\}$ for a fixed $\theta\in\R$, i.e.\ the financial system is in distress if the total losses fall below a certain threshold $\theta$. 
Let $\CG:=\sigma\left(\left\{\cL_{\text{loss}}(X)\leq \theta\right\}\right)$. 
As a CBRM choose $\rGn(\eqF)=\BEW{\PWQ}{-\eqF}{\CG}$, where $\PWQ$ is a risk neutral measure which is equivalent to $\PW$.
The resulting positive homogeneous and convex CSRM evaluated in $\left\{\cL_{\text{loss}}(X)\leq \theta\right\}$ is given by
$$\BEW{\PWQ}{\sum_{i=1}^d \eqY_i^-}{\cL_{\text{loss}}(X)\leq \theta},\quad \eqY\in\Linf{\CF}{d},$$
which corresponds to the DIP for $\eqY=\eqX$. Since the expectation is under a risk neutral measure it can be interpreted as the premium of an aggregate excess loss reinsurance contract. 
\end{Example}

\begin{Example}[Contagion model]\th\label{EX:CM}
In this example we want to specify an aggregation function that explicitly models the default mechanisms in a financial system and perform a small simulation study.
For this purpose we will assume the simplified balance sheet structure given in Table \ref{tbl:balancesheet} for each of the $d$ financial institutions. Let $\eqX\in\Linf{\CF}{d}$ be the vector of equity values of the financial institutions after some market shock on the external assets/liabilities. 
Moreover let $\BPi$ be the relative liability matrix of size $d\times d$, i.e.\ the $i,j$th entry represents the proportion of the total interbank liabilities of institution $i$ which it owes to institution $j$. We denote the $d$-dimensional vector of the total interbank liabilities by $L$.
\begin{table}[htb!]
	\begin{center}
		\begin{tabular}{|c|c|}
		\hline
		Assets&Liabilities\\\hline\hline
		\multirow{2}{*}{External Assets}&Equity\\\cline{2-2}
		&\multirow{2}{*}{External Liabilities}\\\cline{1-1}
		\multirow{2}{*}{Interbank Assets}&\\\cline{2-2}
		&Interbank Liabilities\\\hline
		\end{tabular}
		\caption{Stylized balance sheet.}
		\label{tbl:balancesheet}
	\end{center}
\end{table}

We now consider an extension of the aggregation function proposed by \cite{Chen2013} which is based on the default model in \cite{Eisenberg2001}:\\
For a deterministic vector of equity values $x\in\R^d$ we define the DAF $\dL_{\text{CM1}}$ by the optimization problem:
\begin{align}
\dL_{\text{CM1}}(x):=\max_{y,b\in\R^d_+} \quad&\sum_{i=1}^d-\left(x_i+b_i-(\BPi^\top y)_i\right)^- -\gamma b_i\label{eq:ENobj}\\
\text{subject to}\quad&y=\max\left(\min\left(\BPi^\top y-x-b,L\right),0\right),\label{eq:ENconstraint}
\end{align}
where $y_i$ is the amount by which financial institution $i$ decreases its total liabilities to the remaining institutions and $b\in\R^d$ represents the option of an external participant, e.g.\ a lender of last resort, to inject a capital amount $b_i$ into institution $i$. 
The cost of the injected capital of the lender of last resort is modeled by the parameter $\gamma>1$.\\
There are two possible ways a financial institution can default: First it might default due to the market shock right at the beginning ($x_i<0$). Secondly, if it still has sufficient capital endowment after the market shock, the losses from other institutions might force it into default by contagion effects ($x_i-(\BPi^\top y)_i<0$).
The constraint \eqref{eq:ENconstraint} expresses that if a financial institution defaults, it can either reduce its payments to other institutions or the lender of last resort has to inject capital to cover the default losses. As opposed to the framework in \cite{Chen2013} we are able to incorporate the limited liability assumption ($y\leq L$) proposed in \cite{Eisenberg2001}. Furthermore the lender of last resort will only inject capital into a financial institution as long as the benefit from preventing further contagion exceeds the costs of the injection of the lender of last resort.\\
It can be readily seen that $\dL_{\text{CM1}}$ is isotone and continuous.
The aggregation function $\dL_{\text{CM1}}$ given in \eqref{eq:ENobj} is deterministic. One possible extension within our framework is now to consider conditional modifications of $\dL_{\text{CM1}}$. For example, if there exists only partial information or uncertainty about the future of the interbank liability structure then the relative liability matrix $\BPi(\omega)$ and/or the total interbank liabilities $L(\omega)$ might be modeled stochastically. In this case it can be easily seen that the corresponding aggregation function is a CAF.


We will complete this example by employing the aggregation function $\dL_{\text{CM1}}$ in a small simulation study. The simulation serves illustration purposes only and does not have the objective to represent a real world financial system.
We begin with the construction of network with $10$ institutions as a realization of an Erd\"{o}s-R\'enyi graph with success probability $p=0.35$, that is there exists a directed edge between institution $i$ and $j$ with a probability of $35\%$ independent of the other connections. Furthermore we assume that the exposures between financial institutions follow a half-normal distribution. 
So far we have only knowledge about the size of the interbank assets/liabilities. For the remaining parts of the balance sheet (see Table \ref{tbl:balancesheet}) we assume that the value of equity is a fixed proportion of the total assets and that the external assets/liabilities are chosen such that the balance sheet balances out.
The resulting financial system can be found in Figure \ref{fig:financialsystem}.\\
In the following we want to investigate the impact on the financial system if the institutions are exposed to a shock on their external, i.e. non-interbank, assets and liabilities. For this purpose we add a shock to the initial equity which is normally distributed with mean zero and a standard deviation which is proportional to the financial institutions external assets/liabilities. The single shocks are positively correlated with $\rho=0.1$. \\
In Table \ref{tbl:EN} we list some comparative statistics of the financial system for 30'000 shock scenarios and for different costs of the regulator. The first two rows consider the CSRM's obtained by composing the aggregation function $\dL_{\text{CM1}}$ with the negative expectation and the VaR at level 5\%, resp. Note that we also included the asymptotic case of $\gamma\to\infty$, which corresponds to the situation in which the regulator does not intervene.

\begin{figure}[htb!]
		\centering
		\includegraphics[width=14cm]{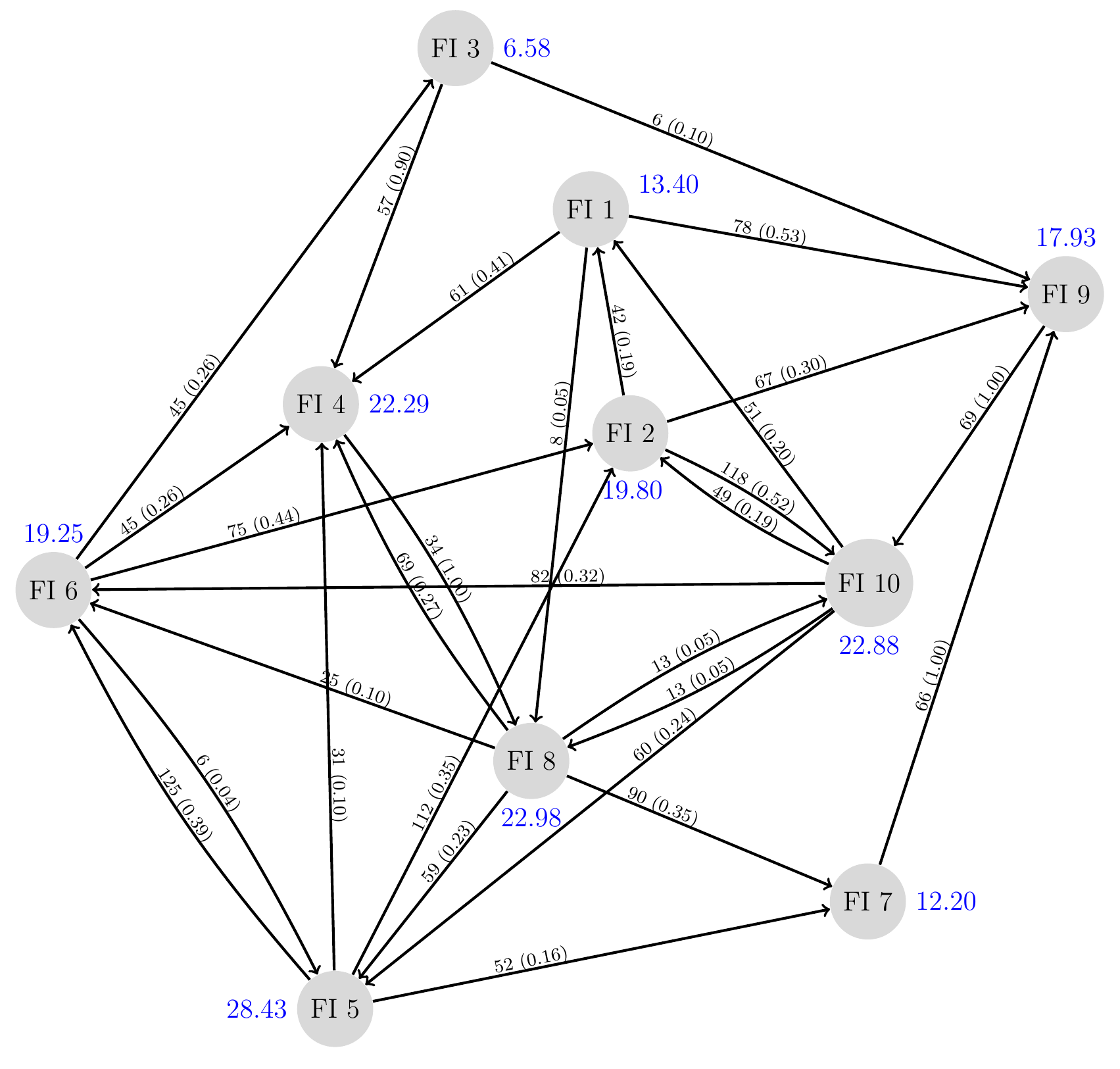}
		\caption{Exemplary financial system.}
		\label{fig:financialsystem}
\end{figure}

\begin{table}[htb!]
	\begin{center}
		\begin{tabular}{ |l|c|c|c| }
			\hline
			$\gamma $						&1.6			&2.6				&'$\infty$'		\\ \hline
			$-\EW{}{\cL_{\text{CM1}}(X)}$&70.62		&88.00&109.30\\ \hline
			$\VaR_{0.05}(	\cL_{\text{CM1}}(X))$			&213.34		&291.59&442.45\\ \hline
			$\sum b_i	$					&23.01		&10.67&0.00\\ \hline
			$\sum x_i^-$ &\multicolumn{3}{c|}{52.33}\\ \hline
			Initially defaulted banks &\multicolumn{3}{c|}{2.57}\\ \hline
			Defaulted banks after contagion&2.87&3.25&3.58\\ \hline
		\end{tabular}
		\caption{Statistics of the financial system for 30'000 shock scenarios.}
		\label{tbl:EN}
	\end{center}
\end{table}

\noindent We observe that with an increasing $\gamma$ the regulator is less willing to inject capital and thus the contagion effects increase which results in a higher risk in terms of the expectation and the Value at Risk. Moreover without a regulator on average round about one financial institution defaults due to contagion effects.

In the next step we want to investigate the systemic importance of the single institutions. For this purpose we modify the CoVaR in \thref{ex:COVAR}, that is, instead of the summing the losses we use the more realistic CAF $\dL_{\text{CM1}}$. Thus we define for a $q\in(0,1)$:
$$\text{CoVaR}_q^j:=\VaR_q\left(\left.\cL_{\text{CM2}}(\eqX)\right|X_j\leq-\VaR_q(\eqX_j)\right),\quad j=1,...,d,$$
where
\begin{align*}
		\dL_{\text{CM2}}(x):=\max_{y,b\in\R^d_+}				\quad&\sum_{i=1}^d-y_i -\gamma b_i\\
		\text{subject to}			\quad&y=\max\left(\min\left(\BPi^\top y-x-b,L\right),0\right).
\end{align*}
The difference between $\dL_{\text{CM2}}$ and $\dL_{\text{CM1}}$ is that losses in case of a default are only taken into consideration up to the total interbank liabilities of this institution, i.e.\ only the losses which spread into the system are taken into account. For example consider an isolated institution in the system which has a huge exposure to the outside of the system, then in order to identify systemically relevant institution it is not meaningful to aggregate the losses from those exposures, nevertheless from the perspective of the total risk of the system those losses should also contribute as it was done in our prior study. As for $\dL_{\text{CM1}}$ it can be easily seen that $\dL_{\text{CM2}}$ is a CAF.
The results for this risk-consistent systemic risk measures $\text{CoVaR}_q^j,j=1,...,d$ can be found in Table \ref{tbl:sysimportance}. 
\begin{table}[htb!]
	\begin{center}
	\resizebox{14.4cm}{!}{
		\renewcommand{\arraystretch}{1.4}
		\begin{tabular}{ |c|l *{10}{|c}| }
			\hline
			\parbox[t]{3mm}{\multirow{2}{*}{\rotatebox[origin=c]{90}{$\gamma=2.6$}}}&FI $j$&2& 3&  6&  4&  7&  1& 10& 9& 5&  8\\\cline{2-12}
			&$\text{CoVaR}_{0.1}^j$&266.94&  297.28&  298.49&  308.61&  320.58&  322.56&  332.94&  355.23&  362.27&  367.68\\\hline
			\parbox[t]{1.7mm}{\multirow{2}{*}{\rotatebox[origin=c]{90}{$\gamma=\infty$}}}&FI $j$&2& 4& 3& 7& 9& 6& 1&10& 8& 5\\\cline{2-12}
			&$\text{CoVaR}_{0.1}^j$&397.73&  419.11&  423.18&  459.33&  471.81&  473.61&  481.40&  548.21&  563.60&  601.09\\\hline
			\multicolumn{1}{c|}{}&FI $j$& 2& 6&10& 3& 1& 7& 5& 9& 8& 4\\\cline{2-12}
			\multicolumn{1}{c|}{}&-$\VaR_{0.1}(\eqX_j)$&	13.30&   -7.67&  -15.05&  -17.01&  -20.69&  -22.98&  -26.89&  -30.48&  -32.11&  -33.41\\\cline{2-12}
			\multicolumn{1}{c|}{}&FI $j$&  4&3&7&9&1&6&2&10&8&5\\\cline{2-12}
			\multicolumn{1}{c|}{}&$L_j$&	34&63&66&69&147&171&227&255&256&320\\\cline{2-12}
		\end{tabular}
		\renewcommand{\arraystretch}{1}
		}
		\caption{Systemic importance ranking based on $\text{CoVaR}_{0.1}^j$.}
		\label{tbl:sysimportance}
	\end{center}
\end{table}
We observe that the systemic importance is always a trade-off between the possibility of high downward shocks and the ability to transmit them. For instance institution 2 can transfer losses up to 227, but it is also the institution which is the least exposed to the market, which makes it also the least systemic important institution. Contrarily institution 4 is the most exposed institution, but does not have the ability to transmit those losses which also results in a low position in the systemic importance ranking. Finally institution 5 or 8 are very vulnerable to the market and have the largest total interbank liabilities and are thus identified as the most systemic institutions.

\end{Example}
\bibliographystyle{chicago}

\begin{thebibliography}{}
	
	\bibitem[\protect\citeauthoryear{Acciaio and Penner}{Acciaio and
		Penner}{2011}]{Acciaio2011}
	Acciaio, B. and I.~Penner (2011).
	\newblock Dynamic risk measures.
	\newblock In J.~Di~Nunno and B.~{\O}ksendal (Eds.), {\em Advanced Mathematical
		Methods for Finance}, Chapter~1., pp.\  11--44. Springer.
	
	\bibitem[\protect\citeauthoryear{Acharya, Pedersen, Philippon, and
		Richardson}{Acharya et~al.}{2010}]{Acharya2010}
	Acharya, V., L.~Pedersen, T.~Philippon, and M.~Richardson (2010).
	\newblock Measuring systemic risk.
	\newblock Available at SSRN 1573171.
	
	\bibitem[\protect\citeauthoryear{Adrian and Brunnermeier}{Adrian and
		Brunnermeier}{2011}]{Adrian2011}
	Adrian, T. and M.~K. Brunnermeier (2011).
	\newblock Co{V}a{R}.
	\newblock Technical report, National Bureau of Economic Research.
	
	\bibitem[\protect\citeauthoryear{Amini and Minca}{Amini and
		Minca}{2013}]{Amini2013b}
	Amini, H. and A.~Minca (2013).
	\newblock Mathematical modeling of systemic risk.
	\newblock In {\em Advances in Network Analysis and its Applications}, pp.\
	3--26. Springer.
	
	\bibitem[\protect\citeauthoryear{Aubin and Frankowska}{Aubin and
		Frankowska}{2009}]{Aubin2009}
	Aubin, J.-P. and H.~Frankowska (2009).
	\newblock {\em Set-valued analysis}.
	\newblock Springer Science \& Business Media.
	
	\bibitem[\protect\citeauthoryear{Bisias, Flood, Lo, and Valavanis}{Bisias
		et~al.}{2012}]{Bisias2012}
	Bisias, D., M.~Flood, A.~W. Lo, and S.~Valavanis (2012).
	\newblock A survey of systemic risk analytics.
	\newblock {\em Annu. Rev. Financ. Econ.\/}~{\em 4\/}(1), 255--296.
	
	\bibitem[\protect\citeauthoryear{Brunnermeier and Cheridito}{Brunnermeier and
		Cheridito}{2013}]{Brunnermeier2013}
	Brunnermeier, M.~K. and P.~Cheridito (2013).
	\newblock Measuring and allocating systemic risk.
	\newblock Available at SSRN 2372472.
	
	\bibitem[\protect\citeauthoryear{Cerreia-Vioglio, Maccheroni, Marinacci, and
		Montrucchio}{Cerreia-Vioglio et~al.}{2011}]{Cerreia-Vioglio2011}
	Cerreia-Vioglio, S., F.~Maccheroni, M.~Marinacci, and L.~Montrucchio (2011).
	\newblock Risk measures: rationality and diversification.
	\newblock {\em Mathematical Finance\/}~{\em 21\/}(4), 743--774.
	
	\bibitem[\protect\citeauthoryear{Chen, Iyengar, and Moallemi}{Chen
		et~al.}{2013}]{Chen2013}
	Chen, C., G.~Iyengar, and C.~C. Moallemi (2013).
	\newblock An axiomatic approach to systemic risk.
	\newblock {\em Management Science\/}~{\em 59\/}(6), 1373--1388.
	
	\bibitem[\protect\citeauthoryear{Cont, Moussa, and Santos}{Cont
		et~al.}{2013}]{Cont2013}
	Cont, R., A.~Moussa, and E.~B. Santos (2013).
	\newblock Network structure and systemic risk in banking systems.
	\newblock In J.-P. Fouque and J.~A. Langsam (Eds.), {\em Handbook on Systemic
		Risk}, Chapter 13., pp.\  327--368. Cambridge University Press.
	
	\bibitem[\protect\citeauthoryear{Detlefsen and Scandolo}{Detlefsen and
		Scandolo}{2005}]{Detlefsen2005}
	Detlefsen, K. and G.~Scandolo (2005).
	\newblock Conditional and dynamic convex risk measures.
	\newblock {\em Finance and Stochastics\/}~{\em 9\/}(4), 539--561.
	
	\bibitem[\protect\citeauthoryear{Eisenberg and Noe}{Eisenberg and
		Noe}{2001}]{Eisenberg2001}
	Eisenberg, L. and T.~H. Noe (2001).
	\newblock Systemic risk in financial systems.
	\newblock {\em Management Science\/}~{\em 47\/}(2), 236--249.
	
	\bibitem[\protect\citeauthoryear{Engle, Jondeau, and Rockinger}{Engle
		et~al.}{2014}]{Engle2014}
	Engle, R., E.~Jondeau, and M.~Rockinger (2014).
	\newblock Systemic risk in europe.
	\newblock Forthcoming in Review of Finance.
	
	\bibitem[\protect\citeauthoryear{F{\"o}llmer and Kl{\"u}ppelberg}{F{\"o}llmer
		and Kl{\"u}ppelberg}{2014}]{Follmer2014a}
	F{\"o}llmer, H. and C.~Kl{\"u}ppelberg (2014).
	\newblock Spatial risk measures: local specification and boundary risk.
	\newblock In {\em Crisan, D., Hambly, B. and Zariphopoulou, T.: Stochastic
		Analysis and Applications 2014 - In Honour of Terry Lyons}. Springer.
	
	\bibitem[\protect\citeauthoryear{F{\"o}llmer and Schied}{F{\"o}llmer and
		Schied}{2011}]{Follmer2011}
	F{\"o}llmer, H. and A.~Schied (2011).
	\newblock {\em Stochastic Finance: An introduction in discrete time\/} (3rd
	ed.).
	\newblock De Gruyter.
	
	\bibitem[\protect\citeauthoryear{Huang, Zhou, and Zhu}{Huang
		et~al.}{2012}]{Huang2012}
	Huang, X., H.~Zhou, and H.~Zhu (2012).
	\newblock Systemic risk contributions.
	\newblock {\em Journal of financial services research\/}~{\em 42\/}(1-2),
	55--83.
	
	\bibitem[\protect\citeauthoryear{Kromer, Overbeck, and Zilch}{Kromer
		et~al.}{2013}]{Kromer2013}
	Kromer, E., L.~Overbeck, and K.~Zilch (2013).
	\newblock Systemic risk measures on general probability spaces.
	\newblock Available at SSRN 2268105.
	
	\bibitem[\protect\citeauthoryear{Revuz and Yor}{Revuz and
		Yor}{1999}]{Revuz1999}
	Revuz, D. and M.~Yor (1999).
	\newblock {\em Continuous martingales and Brownian motion}, Volume 293.
	\newblock Springer.
	
\end{thebibliography}

\end{document}